\documentclass[lettersize,journal]{IEEEtran}
\usepackage{amsmath,amsfonts}
\usepackage{array}
\usepackage[caption=false,font=normalsize,labelfont=sf,textfont=sf]{subfig}
\usepackage{textcomp}
\usepackage{stfloats}
\usepackage{url}
\usepackage{verbatim}
\usepackage{graphicx}
\usepackage{cite}
\DeclareMathOperator{\lcm}{lcm}
\usepackage[procnumbered,linesnumbered,ruled,vlined]{algorithm2e}
\usepackage{algpseudocode}

\SetKwFor{for}{for each}{par-do}{endfor}
\usepackage{color}
\usepackage{amssymb}
\usepackage{amsthm}
\usepackage{makecell}
\usepackage{multirow}

\begin{document}

\title{Scalable Scheduling for Industrial Time-Sensitive Networking: A Hyper-flow Graph Based Scheme}

\author{Yanzhou~Zhang,
	Cailian~Chen,~\IEEEmembership{Member,~IEEE,}
	Qimin~Xu,~\IEEEmembership{Member,~IEEE,}
	Shouliang~Wang,
	Lei~Xu,
	and~Xinping~Guan,~\IEEEmembership{Fellow,~IEEE}
}

\markboth{Journal of \LaTeX\ Class Files,~Vol.~14, No.~8, August~2021}%
{Shell \MakeLowercase{\textit{et al.}}: A Sample Article Using IEEEtran.cls for IEEE Journals}

\IEEEpubid{0000--0000/00\$00.00~\copyright~2021 IEEE}

\maketitle

\begin{abstract}
Industrial Time-Sensitive Networking (TSN) provides deterministic mechanisms for real-time and reliable flow transmission. Increasing attention has been paid to efficient scheduling for time-sensitive flows with stringent requirements such as ultra-low latency and jitter. In TSN, the fine-grained traffic shaping protocol, cyclic queuing and forwarding (CQF), eliminates uncertain delay and frame loss by cyclic traffic forwarding and queuing. However, it inevitably causes high scheduling complexity. Moreover, complexity is quite sensitive to flow attributes and network scale. The problem stems in part from the lack of an attribute mining mechanism in existing frame-based scheduling. For time-critical industrial networks with large-scale complex flows, a so-called hyper-flow graph based scheduling scheme is proposed to improve the scheduling scalability in terms of schedulability, scheduling efficiency and latency \& jitter. The hyper-flow graph is built by aggregating similar flow sets as hyper-flow nodes and designing a hierarchical scheduling framework. The flow attribute-sensitive scheduling information is embedded into the condensed maximal cliques, and reverse maps them precisely to congestion flow portions for re-scheduling. Its parallel scheduling reduces network scale induced complexity. Further, this scheme is designed in its entirety as a comprehensive scheduling algorithm GH$^2$. It improves the three criteria of scalability along a Pareto front. Extensive simulation studies demonstrate its superiority. Notably, GH$^2$ is verified its scheduling stability with a runtime of less than 100 ms for 1000 flows and near 1/430 of the SOTA FITS method for 2000 flows.
\end{abstract}

\begin{IEEEkeywords}
Industrial Time-Sensitive Networking, scalable scheduling, hyper-flow graph, hierarchical framework
\end{IEEEkeywords}

\section{Introduction}
\IEEEPARstart{A}{s} a basic infrastructure of the Industrial Internet of Things, device networking in factory automation is incorporating more and more sensors, controllers and actuators to build ubiquitous networked cyber-physical systems. Massive network applications require ultra-low latency and jitter, the most critical quality of service (QoS) metrics. For example, industrial automation and power grid systems require tight deterministic latency in the range of milliseconds to microseconds \cite{8458130}. The jitter is also bounded within a few microseconds \cite{7883994}. These drive the industrial network evolution to ultra-low latency and ultra-high reliability.

Time-Sensitive Networking (TSN), an Ethernet extension under development from IEEE 802.1\texttrademark\ TSN Task Group, empowers Ethernet bridge networks determinism by a capability collection governed via a series of standards. It integrates mechanisms such as traffic shaping, bandwidth policing and clock synchronization to control flow forwarding in a precise manner with bounded low latency \& jitter and extremely low frame loss. Specifically, IEEE 802.1Qbv \cite{8613095} builds a gate-operation mechanism named time-aware shaper during the forwarding process. Inside the output port of TSN supported network devices, this shaper works to manage traffic forwarding queues and time through a gate control list (GCL). It records sequential gate states and state durations. Following this, TSN devices transmit flows in queues orderly and cyclically. Similarly, IEEE 802.1Qci \cite{8064221} constructs an en-queue policing mechanism for flow caching inside every TSN device. The flow forwarding between these devices is synchronized with IEEE 802.1AS \cite{9121845}. Moreover, a ping-pong gate configuration model, called cyclic queuing and forwarding (CQF), is given as IEEE 802.1Qch \cite{7961303} to implement the transmission determinism with these mechanisms. Nevertheless, further flow scheduling is desired for the real-time transmission of massive flows with ultra-low latency and jitter limits. It is considered as the deterministic scheduling problem (DSP).
\IEEEpubidadjcol

Since the advent of TSN, related research for DSP has been widely conducted and discussed. This problem is first formalized with systematic constraints \& goals and proved to be NP-hard \cite{2997465}. It is solved by Satisfiability Model Theory \cite{8430062} or Integer Linear Programming \cite{8607243}, which show a high complexity and are hardly adaptable to large-scale complex industrial flow scenarios. Thus, the rapid schedulability assessment \cite{3356401, 8757948}, and even further, the efficient determinism scheme design is desired to achieve the scalable scheduling capability. With such a vision, the flow incremental or partitioning scheduling methods are proposed in works \cite{8186237, 9407828, 9893358, 9684566, 8889667}. Typically, Quan \textit{et al.} \cite{9407828} designs a flow injection-time scheduling (FITS) method that allocates flows incrementally with reasonable slot resources by injection offset. Guo \textit{et al.} \cite{9893358} co-designs the flow incremental ordering and injection offset selection by a mapping score metric-based scheduling (MSS) method. In addition, Atallah \textit{et al.} \cite{8889667} proposes an iterated flow partition scheduling method, where a degree of conflict-aware stream partitioning (DASP) strategy is adopted to optimize scheduling partitioning. Compared to global scheduling, these above ways segment flows into serial and interrelated scheduling blocks to downsize the network scale of each scheduling step. It facilitates the balance between scheduling optimality and efficiency. However, the complexity of these serial patterns is still high and difficult to improve. And even more so, their frame-based patterns during each block scheduling lack in attribute mining mechanism that inevitably possesses flow attribute induced high complexity.

To further cope with the complexity dilemma, the correlation between TSN mechanisms and flow attributes \cite{9809824, 9812895, 9714183} is deeply analyzed to optimize the key scheduling processes, like constraint satisfiability and goal optimality checking. A divisibility theory-based flow sequence analysis method \cite{9714183} is developed to metric the scheduling space and elaborate flow confluences on time slots under the CQF model. Its building flow graph converts the frame-based scheduling into the equivalent flow graph-based pattern for efficient scheduling. Even so, its flow scale induced exponential graph complexity restricts the scheduling efficiency. Therefore, the current DSP is desired to schedule with better scalability, which means the faster scheduling with superior schedulability and latency \& jitter performances.

This paper focuses on improving the scheduling efficiency and scalability under the QoS and determinism assurances for large-scale complex flow scenarios. With the CQF-configured network, we holistically consider the scheduling complexity and optimality and deeply mine flow attribute-driven scheduling features. A hyper-flow graph based scheme is developed to innovate traditional serial and frame-based scheduling patterns. With this scheme, the critical causes of scheduling complexity, network scale (including device and flow scale) and flow attributes, are significantly suppressed in their impact on scheduling. Meanwhile, the schedulability and QoS performances are well-guaranteed and even enhanced. The contributions of the scheme are summarized as follows:
\begin{itemize}
\item[$ \bullet $] A hyper-flow graph based methodology is proposed to design and optimize the flow scheduling processes systematically. Instead of traditional frame-based patterns, the attribute-driven hyper-flow graph maps/embeds their flow attribute-sensitive redundant scheduling information into equivalent and less maximal hyper-flow cliques. Also, the re-scheduling information is precisely reverse mapped without violent retrieve by introducing the conflict clique. These ways suppress flow attribute induced complexity and drive the Pareto-optimized scheduling.

\item[$ \bullet $] A hierarchical scheduling framework is constructed to balance scheduling efficiency and optimality for large-scale complex flows. With the unified design by hyper-flow graphs, four progressive phases are closely interlinked, including lightweight flow partitioning, parallel partition scheduling, parallel flow synthesizing and precise flow re-scheduling. These two parallel phases reduce the flow and device scale induced scheduling complexity, respectively, and the last phase improves the schedulability and optimality with precise flow fine-tuning.

\item[$ \bullet $] A hyper-flow graph based hierarchical (GH$^2$) algorithm designs the scheduling scheme in its entirety. Guided by the graph, an attribute-driven partitioning strategy, a parallel hyper-flow graph based scheduling (HFG) method, a parallel Bron-Kerbosch synthesizing method and a conflict clique based re-scheduling (CCR) method are designed comprehensively and interlinked hierarchically. It improves the schedulability, scheduling efficiency and QoS performances along a Pareto front.
\end{itemize}

The rest of this paper is organized as follows. First, Section II establishes the CQF-configured network model and formulates its corresponding scheduling problem. Then, the hyper-flow graph based scheme is illustrated in Section III, including the hyper-flow graph based methodology, hierarchical framework and comprehensive flow scheduling algorithm. Section IV evaluates the performances of GH$^2$, and Section V makes the final conclusion.

\section{Network Model and Problem Formulation}
\subsection{Network and Flow Model}
For the target network with the determinism demands, we model it as a directed graph $ \mathcal{N}=<\mathcal{V,E}> $, where $ \mathcal{V} $ denotes a set of network nodes $ \nu_{\kappa} $, consisting of terminal hosts $ \mathcal{V}^h $ and TSN switches $ \mathcal{V}^s $. $ \mathcal{E} $ denotes the directed physical links $ \varepsilon_\iota $.

In this network, the time-sensitive applications communicate through information flows $ \mathcal{F} $. Each of them $ f_i $ is characterized by the attribute tuple $ \{ l_i, p_i, b_i, d_i, j_i, R_i \} $. These denote the frame length, flow period, baseline time of generation, maximum allowable latency, maximum allowable jitter and transmission route, respectively. The flow route $ R_i $ is predefined as the following form of link association:
\begin{equation}
	R_i = [ \varepsilon_{i, 0}, \varepsilon_{i, 1} \cdots, \varepsilon_{i, h_i} ], 
\nonumber
\end{equation}
which connects the links between the source and destination host nodes with the hop count $ h_i $.

\begin{figure}[!b]
	\centerline{\includegraphics[width=8cm]{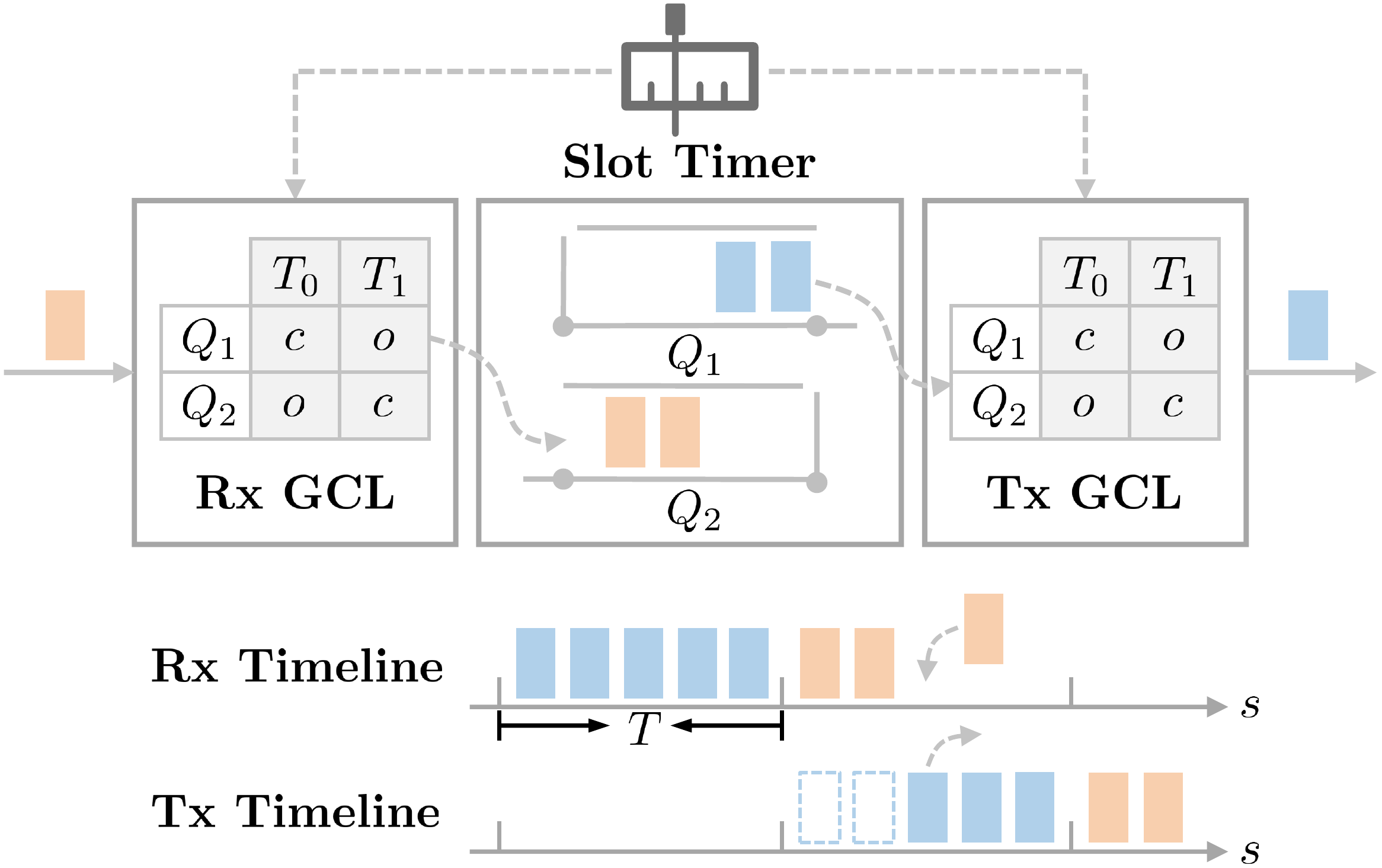}}
	\caption{The CQF model for each output port of TSN switches}
	\label{CQF}
\end{figure}
\begin{figure}[!t]
	\centerline{\includegraphics[width=5.7cm]{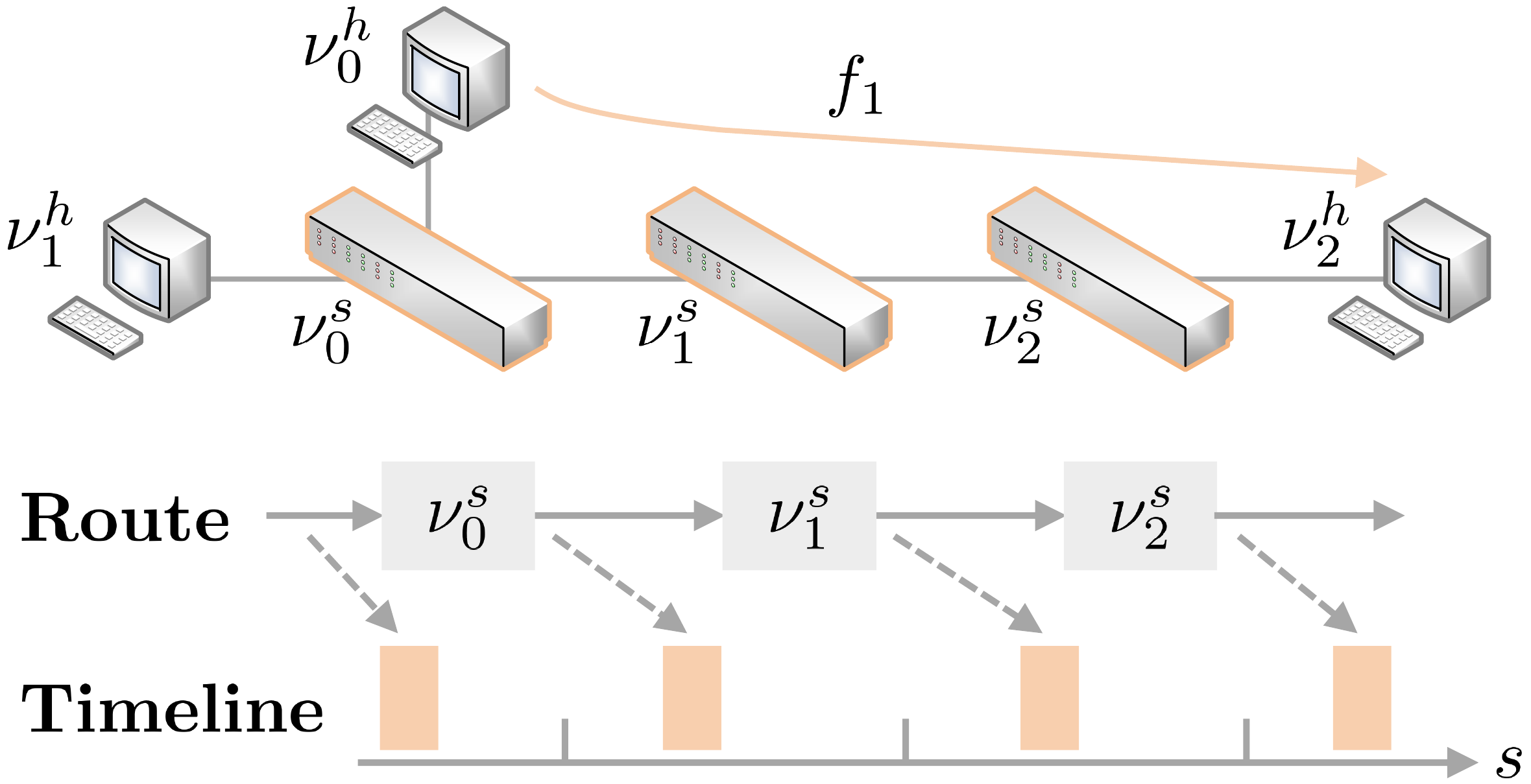}}
	\caption{Slot-by-slot flow transmission under the CQF model}
	\label{CQF2}
\end{figure}
To guarantee the deterministic transmission of target flows, that is, the bounded latency and jitter, all the output ports of TSN switches are implemented with two statically configured Ping-Pong queues from IEEE 802.1Qch, called the CQF model. As shown in Fig. \ref{CQF}, it performs en-queue and de-queue operations alternately with a predetermined duration $ T $. The transmission and reception control between both nodes of link $ \varepsilon_\iota $ are precisely aligned based on IEEE 802.1AS. In this way, the transmission process is divided into numbered time slots $ s \in \mathcal{S} $ with length $ T $, and the target flows are transmitted and queued along their route $ R_i $ slot after slot like Fig. \ref{CQF2}. The end-to-end determinism of flow transmission is assured with the bounded latency
\begin{equation}
	[h_i \cdot T-T, h_i \cdot T+T]
\end{equation}
and jitter $ [0, 2 \cdot T] $. Considering the slotted transmission behavior, partial flow attributes are converted into the slot expression as $ \mathring{p}_i, \mathring{b}_i, \mathring{d}_i, \mathring{j}_i $, given by $ \frac{p_i}{T} $, $ \lfloor \frac{b_i}{T} \rfloor $, $ \lfloor \frac{d_i}{T} \rfloor $, $ \lfloor \frac{j_i}{T} \rfloor $, respectively.

\begin{figure}[!h]
	\centerline{\includegraphics[width=5.7cm]{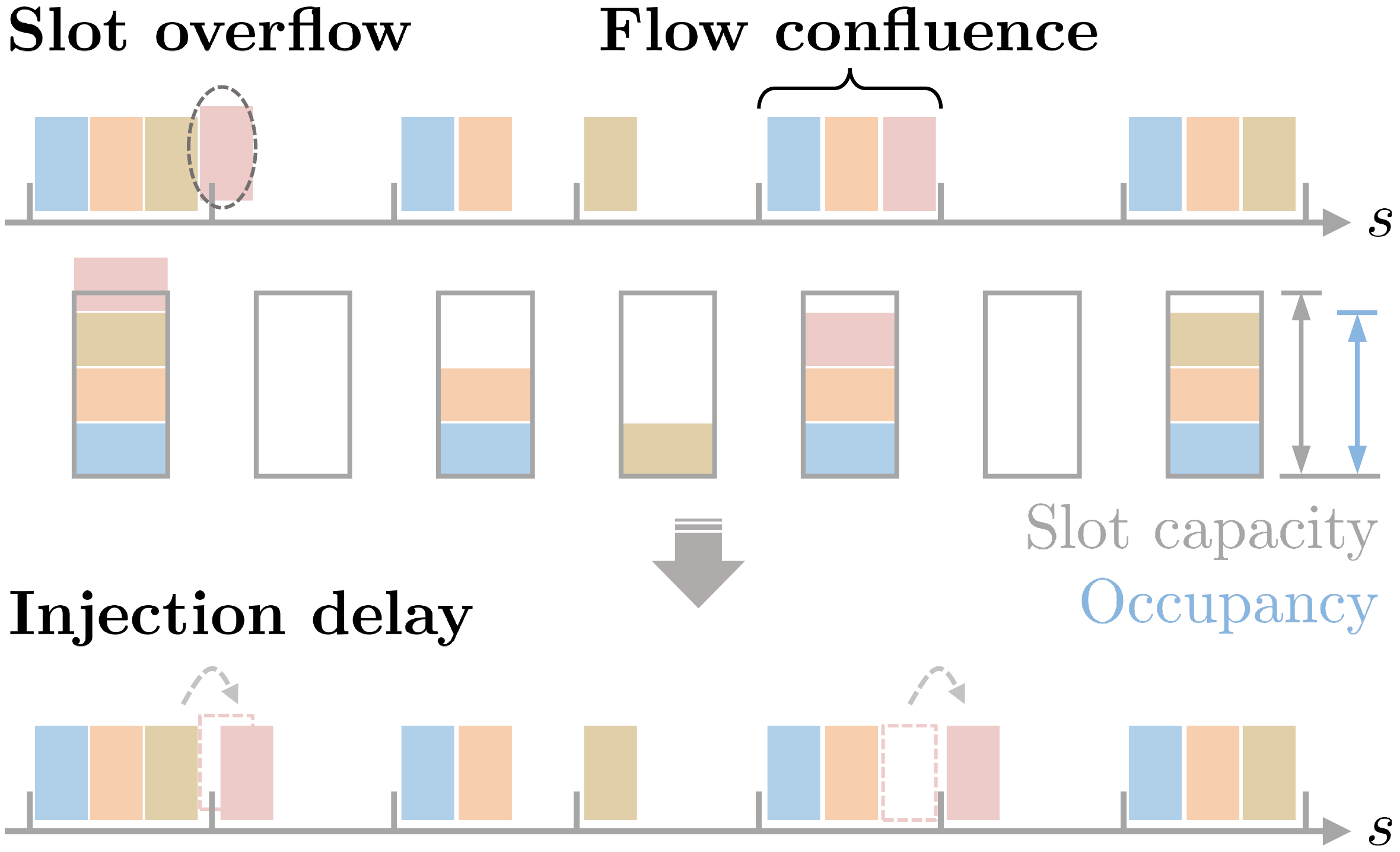}}
	\caption{Injection offset design for CQF-based flow scheduling}
	\label{injection}
\end{figure}
As analyzed in \cite{9155434}, unscheduled flow injection into the network is easy to cause chaotic flow confluences on the same slot. It further results in slot congestion and overflow due to excessive occupancy beyond the slot capacity. Further, the potential frame loss occurs and the deterministic mechanism is broken. To avoid this, the injection delay $ o_i $ is introduced to regulate the postponement between flow generation and injection. As shown in Fig. ref{injection}, the offset design of flow injection allows slot resources to be allocated to target flows without overflow, which facilitates flow scheduling more flexibly and effectively. It is worth noting that the injection offset is symbolized as $ o_i $ under a general assumption \cite{9155434,9407828,9714183,8607243,8757948,8186237,9893358,9684566,8889667,9809824,9812895,10001004,10101832,8700610} that all frames of each flow $ f_i $ share the same delay. This assumption reduces the variable scale to boost the scheduling efficiency and eliminates the additional flow jitter to improve QoS performances, but it is not necessary to obey.

\subsection{Hyper-flow Graph Model}
Given the generation periodicity and transmission determinism, target flows are kept forwarding periodically on every link by the uniform injection offset among each flow. Based on this, we cluster the frames of each flow with a unified forwarding slot representation. For any flow $ f_i \in \mathcal{F} $ on link $ \varepsilon_\iota \in R_i $, its frame forwarding slots are modeled as a periodic sequence
\begin{equation}
	\aleph(k_i^{\varepsilon_\iota}) = q_i^{\varepsilon_\iota} +k_i^{\varepsilon_\iota} \cdot \mathring{p_i},
\end{equation}
where the baseline forwarding slot $ q_i^{\varepsilon_\iota} $ is given as $ \mathring{b_i} +o_i +h_{i}^{\varepsilon_\iota} $. $ h_{i}^{\varepsilon_\iota} $ is the passed hop number among $ R_i $, that is, the passed slots after flow injecting. The sequence count $ k_i $ accumulates every periodic frame. The flow attribute integrated slot positioning parameters $ (q_i^{\varepsilon_\iota}, \mathring{p_i}) $ are called as the feature tuple of flow $ f_i $.

\begin{figure}[!h]
	\centerline{\includegraphics[width=3cm]{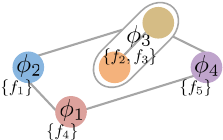}}
	\caption{Schematic of a hyper-flow graph}
	\label{graph}
\end{figure}
In order to efficiently guide flow offset selection, we build a correlation graph on each link $ \varepsilon_\iota \in \mathcal{E} $ to characterize flow confluence relationships. Considering the specific forwarding slots of target flows are dominated by their feature tuples, we aggregate the flow set with the same $ (q_i^{\varepsilon_\iota}, \mathring{p_i}) $ as a hyper-flow $ \phi $ and weight it the frame length sum among the flow set. Then, we take the hyper-flows  $ \Phi^{\varepsilon_\iota} $ as graphical nodes and connect their undirected edges following the adjacency matrix $ \Psi^{\varepsilon_\iota} $ as
\begin{equation}
	\Psi_{\phi_{\alpha}, \phi_{\beta}}^{\varepsilon_\iota} =
	\begin{cases}
		1 & \text{if\ } \gcd(\mathring{p}_{\alpha}, \mathring{p}_{\beta}) | (q_{\alpha}^{\varepsilon_\iota} -q_{\beta}^{\varepsilon_\iota}). \\
		0 & \text{otherwise.}
	\end{cases}
\end{equation}
Like Fig. \ref{graph}, the flow aggregated hyper-flows are correlated by the relationship between their forwarding slot parameters. The hyper-flow graph is desired to elaborate the flow confluences on slots and efficiently extract the key scheduling information.

\subsection{Problem Formulation for CQF-configured network}
With the above models, we construct the optimal scheduling problem for the CQF-configured network with variable flow injection offset. To meet flow QoS demands and network determinism requirements, the CQF-based DSP is demonstrated by the following constraints and goal to regulate the scheduling selection of flow offsets.

\subsubsection{Latency Constraint}
In the CQF model, flow transmission latency is locked as (1) within the switching network. Affected by the injection offset $ o_i $, the worst-case latency $ \vartheta_i $ is extended to $ (o_i +h_i +1) \cdot T $. In order to bound $ o_i $ into a QoS allowed range, we construct the latency constraint as
\begin{equation}
	\begin{aligned}
	& \forall f_i \in \mathcal{F}: \\
	& 0 \leq o_i < \mathring{d}_i -h_i.
	\end{aligned}
\end{equation}

\subsubsection{Jitter Constraint}
Besides the latency, the flow jitter also needs to be within a tolerable limitation $ \mathring{j}_i $ as follows. It is produced by the possible offset fluctuation $ \tilde{o_i} $ within a flow and the endogenous jitter $ 2 \cdot T $ from the CQF model.
\begin{equation}
	\begin{aligned}
	& \forall f_i \in \mathcal{F}: \\
	& \tilde{o_i} + 2 \leq \mathring{j}_i.
	\end{aligned}
\end{equation}

\subsubsection{Capacity Constraint}
During the flow forwarding on the common links, there exit potential flow confluences sharing the same slot $ s^{\varepsilon_\iota} $. Constrained by the limited slot capacity $ \Lambda $, the target flows need to be allocated reasonable slots without overflow. Thus, the constraint is demonstrated as:
\begin{equation}
	\Delta_{i,s}^{\varepsilon_\iota} =
	\begin{cases}
	1 \leftarrow s^{\varepsilon_\iota} \equiv q_i^{\varepsilon_\iota} \bmod \mathring{p}_i, \\
	0 \leftarrow \text{otherwise},
	\end{cases}
\end{equation}
\begin{equation}
	\begin{aligned}
	& \forall \varepsilon_\iota \in \mathcal{E}, s^{\varepsilon_\iota} \in \{1,2,\cdots,C\}: \\
	& \sum_{f_i \in \mathcal{F}^{\varepsilon_\iota}} \Delta_{i,s}^{\varepsilon_\iota} \cdot l_i \leq \Lambda,
	\end{aligned}
\end{equation}
where the boolean variable $ \Delta_{i,s}^{\varepsilon_\iota} $ in (6) works to verify whether flow $ f_i $ is forwarded at slot $ s^{\varepsilon_\iota} $ of link $ \varepsilon_\iota $. It is confirmed by the congruence of checked slot $ s^{\varepsilon_\iota} $ and flow forwarding basetime $ q_i^{\varepsilon_\iota} $ in relation to period $ \mathring{p}_i $. The constraint in (7) restricts the slot occupancy filled by flow confluences on each slot $ s^{\varepsilon_\iota} \in \mathcal{S} $ to no more than the available slot capacity $ \Lambda $. It is given as $ \gamma \cdot \min \{(T -\delta^{\varepsilon_\iota}) \cdot \Gamma, \mho^{\varepsilon_\iota}\} $, where $ \delta^{\varepsilon_\iota} $ is the slot misalignment owing to the clock synchronization errors, $ \Gamma $ is the link bandwidth, $ \mho^{\varepsilon_\iota} $ is the physical queue depth, and $ \gamma $ is the capacity distribution factor used to reserve resources for deterministic transmission of aperiodic time-sensitive flows. Moreover, the capacity constraint covers slots $ \mathcal{S} $ within a bounded scheduling view $ C $ given as the least common multiple of all flow periods, that is, $ \lcm(\mathring{P}^n) $, where $ \mathring{P}^n $ is $ \{ \mathring{p}_1, \mathring{p}_2, \cdots, \mathring{p}_n \} $. With this range, the global determinism of flow scheduling is assured by cyclically repeating the flow forwarding in these $ C $ slots.

\begin{figure*}[!ht]
	\centerline{\includegraphics[width=18.2cm]{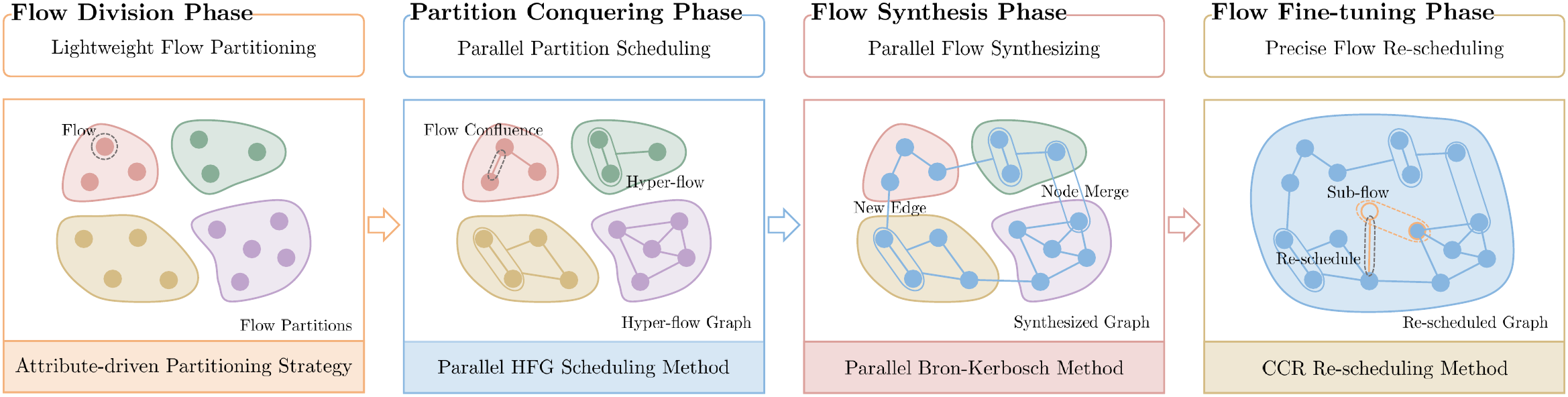}} 
	\caption{Overview of the hyper-flow graph based scheduling scheme}
	\label{method}
\end{figure*}

Among the feasible scheduling space for the above constraints (4)-(7), we aim to enhance the real-time transmission ability of flows $ \mathcal{F} $ and the load balancing ability of slots $ \mathcal{S} $. The former improves latency performances, and the latter helps cope with burst flows and schedulability. Thus, a composite optimization goal is expressed as
\begin{equation}
	\min_{\{o_i\}}{ (1-\rho) \cdot \frac{1}{n} \sum_{f_i}{\frac{\vartheta_i}{\mathring{d}_i}} + \rho \cdot \left[ \frac{\zeta_s^{\varepsilon_\iota}}{\Lambda} \right]^{\max_{s, \varepsilon_\iota}} },
\end{equation}
where the optimization value symbolized as $ \Im $ is the combined weight of the former transmission average real-time rate $ \vartheta_i / \mathring{d}_i $ in target flows and the latter maximum slot occupancy rate $ \zeta_s^{\varepsilon_\iota} / \Lambda $ among all slots on links $ \mathcal{E} $. $ \zeta_s^{\varepsilon_\iota} $ is denoted as the slot occupancy of slot $ s^{\varepsilon_\iota} $ on link $ \varepsilon_\iota $, and $ \rho \in [0, 1] $ is the weight factor working to balance these two goals.

The complexity of general constraint-based optimization problems relies on intensive computation for constraint satisfiability and goal optimality. Considering the formulation of the above CQF-based DSP, the scale of scheduling variables, constraints and goal components expands with the growing network scale including device and flow scale. It causes an increasingly high scheduling complexity. So do the complicating flow attributes, especially the flow periods. According to the TSN profile standard for industrial automation \cite{9714184}, the flow periods of time-sensitive network systems are complex with large range spans. It inflates the scheduling view $ C $ to increase the scheduling complexity. So far, these complexity dilemmas still need to be properly addressed. Thus, we work on exploring a systematic scheme to solve these problems uniformly while well-balancing the scheduling optimality (8), that is, the proposed hyper-flow graph based scheme.

\section{Hyper-flow Graph Based Scheme}
In this section, we elaborate on the hyper-flow graph based scheme including the hyper-flow graph based methodology and four-phase hierarchical framework. It is graphically depicted as Fig. \ref{method}. A comprehensive scheduling algorithm GH$^2$ designs the scheme in its entirety for scalable flow scheduling in CQF-configured TSN. Thus, we build the scheme following the GH$^2$ as Algo. 1 hierarchically.

\begin{algorithm}[!t]
	\caption{GH$^2$ Scheduling Method}
	\KwIn{Target scheduling flows: $ f_i \in \mathcal{F} $ \\
		Flow attributes: $ \{ l_i, p_i, b_i, d_i, j_i, R_i \} \in \mathcal{A} $ \\
		Available slot capacity: $ \Lambda $
	}
	\KwOut{Flow scheduling solution: $ (f_i, o_i) \in \mathcal{O} $
	}
	$ \{ \mathcal{F}_{\varsigma} \} \leftarrow $ Attribute-driven flow partitioning for $ \mathcal{F} $\;
	\If{$ \varsigma > 1 $}
	{
		$ \mathcal{O}, \Phi, \mathcal{L} \leftarrow $ Parallel HFG scheduling($ \{ \mathcal{F}_{\varsigma} \}, \mathcal{A} $)\;
		$ \Upsilon \leftarrow $ Parallel Bron-Kerbosch synthesizing($ \Phi, \mathcal{L} $)\;
	}
	\Else
	{
		$ \mathcal{O}, \Upsilon \leftarrow $ HFG scheduling($ \mathcal{F}, \mathcal{A} $)\;
	}
	$ \dot{\mathcal{O}} \leftarrow $ Precise CCR re-scheduling($ \Upsilon, \Lambda $)\;
	$ \mathcal{O} \leftarrow \mathcal{O} \cup \dot{\mathcal{O}} $\;
\end{algorithm}

Given the complexity analysis of the above DSP, we first explore a divide-and-conquer hierarchical design to reduce the network scale induced complexity. Specifically, the first flow division phase divides target flows $ \mathcal{F} $ into separate partitions $ \mathcal{F}_{\varsigma} $ in line 1, where an attribute $ \mathcal{A} $-driven lightweight strategy is designed to accelerate flow partitioning. For more than one flow partition, the second partition conquering phase spreads them over multiple parallel scheduling sub-problems to reduce the flow scale induced complexity. With the flow scheduling of optimized slot allocations, every partition obtains a local flow offset solution that is aggregated as a preliminary scheduling solution $ \mathcal{O} $ in line 3. Considering the partitions share a common network, the global scheduling information, that is, the slot occupancy by flow confluences, needs to be synthesized from these sub-problems. The third flow synthesis phase works on it link-parallelly in line 4 to suppress the device scale induced complexity. Unusually, when only one partition exists, the synthesis phase is skipped to line 6. For the specific scheduling and synthesizing processes, the hyper-flow graph based methodology is explored by building a hyper-flows $ \Phi $ and weights $ \mathcal{L} $ based graph. It maps/embeds one-by-one slot occupancy within flow attribute-sensitive bound $ C $ into equivalent and less maximal hyper-flow cliques $ \Upsilon $. With the condensed scheduling information, a parallel HFG flow scheduling method in line 3\&6 and a parallel Bron-Kerbosch flow synthesizing method in line 4 are proposed to improve the scheduling efficiency further. By these three phases, GH$^2$ fast-tracks a preliminary scheduling solution and its global maximal hyper-flow cliques under the optimized QoS performances.

After the flow synthesis of scheduling information, there may be some overflow slots where their occupancy exceeds the slot capacity $ \Lambda $ or other desired values. It violates the capacity constraint (7) and affects the scheduling optimality (8). Hence, the last flow fine-tuning phase in line 7 is desired to re-schedule the overflow flows inside these slots. Meanwhile, the above preliminary solution is rectified as slightly as possible to maintain QoS performances. For this vision, the assumption that all frames of each flow share the same delay is broken. We aim to precisely position the overflow slots and re-scheduling overflow flow portions inside them. By extending the hyper-flow graph based methodology, the conflict clique is first introduced and proved equivalent to overflow slots. Then, the re-scheduling information, that is, the overflow slots and their flow portions, is precisely reverse mapped from the maximal hyper-flow cliques without violent retrieve. It drives an efficient CCR re-scheduling method in line 7 to precisely fine-tune flows for schedulability improvement and load-balancing. After this phase, the supplement solution $ \dot{\mathcal{O}} $ is given and constitutes the complete scheduling solution together with $ \mathcal{O} $ in line 8.

The specific processes of the above hyper-flow graph based scheme, that is, GH$^2$ scheduling algorithm, are demonstrated in phases as follows.

\subsection{Lightweight Flow Division Phase}
For the DSP with large-scale complex flows, we divide the target flows $ \mathcal{F} $ into multiple partitions $  \mathcal{F}_{\varsigma} $ for the subsequent separate scheduling. Distinguishing from SOTA partitioning strategies, like spectral clustering \cite{8889667, 9684566} that constructs a similarity graph to cluster flows, we design a lightweight strategy to reduce time costs and maintain scheduling performances. It partitions flows via their sorted typical attitudes.

Considering that load balancing for every partition scheduling sub-problem helps to reduce the slot overflow, we view this as the partition principle of the CQF-based DSP. Then, an intuitive law is observed that the flows with similar attributes facilitate the balance of slot occupancy. For instance, the flows with similar even the same flow period $ p_i $ could be evenly distributed over different slots with fewer flow confluences. The similar frame length $ \mathring{l}_i $ could bring approximate occupancy to balance the slots. Thus, we divide flows based on the similarity of flow attributes as follows.
\begin{equation}
	\begin{aligned}
	& \Theta \leftarrow \{ f_{\alpha}, f_{\beta}, \cdots \in \mathcal{F} \mid f_{\alpha} \prec f_{\beta} \prec \cdots \}, \\
	& \{ \mathcal{F}_1, \mathcal{F}_2, ..., \mathcal{F}_{\varsigma}, ... \} \leftarrow \text{sequentially slice } \Theta,
	\end{aligned}
\end{equation}
where we refer flow attitude $ l_i $, $ p_i $ or the bandwidth consumption depicted by $ l_i / p_i $ as a potential partitioning bases. Target flows are first sorted by one of the above basis and then batch sliced into one-by-one partition $ \mathcal{F}_{\varsigma} $ with a specified scale $ \Xi $.

\subsection{Parallel Partition Conquering Phase}
For the above flow partitions, we schedule them as separate sub-problems to cut their interlink. This way reduces the flow scale in each sub-problem and the scheduling complexity. Due to the separability, a parallel computing structure is adopted to further reduce the complexity compared to the serial patterns. Each conquering unit corresponds to one flow partition and deploys the HFG scheduling method.

\begin{algorithm}[!t]
	\caption{Parallel HFG Scheduling Method}
	\KwIn{Flow partitions: $ \mathcal{F}_{\varsigma} \in \mathcal{F} $ \\
		Flow attributes: $ \{l_i, \mathring{p_i}, \mathring{b_i}, \mathring{d_i}, R_i\} \in \mathcal{A} $
	}
	\KwOut{Flow scheduling solution: $ (f_i, o_i) \in \mathcal{O} $
	}
	\for{$ \mathcal{F}_{\varsigma} $ in $ \mathcal{F} $}
	{	
		$ \Theta_{\varsigma} \leftarrow \{ f_{\alpha}, f_{\beta}, \cdots \in \mathcal{F}_{\varsigma} \mid f_{\alpha} \prec f_{\beta} \prec \cdots \} $\;
		\For{$ f_{\epsilon} $ in $ \Theta_{\varsigma} $}
		{	
			$ o \leftarrow 0 $, $ \bar{o}_{\epsilon} \leftarrow \min \{ p_{\epsilon}, d_{\epsilon}-h_{\epsilon} \} $ and $ \Im_{\epsilon} \leftarrow \infty $\;
			\While{$ o < \bar{o}_{\epsilon} $}
			{
				\For{$ \varepsilon_\iota $ in $ R_{\epsilon} $}
				{
					$ \phi^{\varepsilon_\iota} \leftarrow $ get $ (\mathring{p_{\epsilon}}, q_{\epsilon}^{\varepsilon_\iota}) $ by (2)\;
					\If{$ \phi^{\varepsilon_\iota} $ in $ \Phi_{\varsigma}^{\varepsilon_\iota} $}
					{
						\For{$ \varpi_{\kappa}^{\varepsilon_\iota} $ in $ \Upsilon_{\phi}^{\varepsilon_\iota} $}
						{
							$ \zeta_{\varpi}^{\varepsilon_\iota} \leftarrow l_{\epsilon} + \sum_{\phi \in \varpi_{\kappa}^{\varepsilon_\iota}} \ell_{\phi}^{\varepsilon_\iota} $\;
						}
					}
					\Else
					{
						$ \varkappa_{\phi}^{\varepsilon_\iota} \leftarrow $ get neighbors of $ \phi^{\varepsilon_\iota} $ by (3)\;
						$ \tilde{\Upsilon}_{\phi}^{\varepsilon_\iota} \leftarrow $ get cliques for $ \phi^{\varepsilon_\iota} $ by (12)\;
						
						\For{$ \tilde{\varpi}_{\kappa}^{\varepsilon_\iota} $ in $ \tilde{\Upsilon}_{\phi}^{\varepsilon_\iota} $}
						{
							$ \zeta_{\varpi}^{\varepsilon_\iota} \leftarrow \sum_{\phi \in \tilde{\varpi}_{\kappa}^{\varepsilon_\iota}} \ell_{\phi}^{\varepsilon_\iota} $\;
						}
					}
				}
				$ \bar{\zeta}_o \leftarrow \max{\{ \max{\{ \zeta_{\varpi}^{\varepsilon_\iota} \}, \bar{\zeta}} \} }$\;
				$ \Im_o \leftarrow $ get the optimization value by (11)\;
				\If{$ \Im_o < \Im_{\epsilon} $}
				{
					$ o_{\epsilon} \leftarrow o $, $ \bar{\zeta}_{\epsilon} \leftarrow \bar{\zeta}_o $ and $ \Im_{\epsilon} \leftarrow \Im_o $\;
					$ \{ \phi^{\varepsilon_\iota} \}_{\epsilon} \leftarrow \{ \phi^{\varepsilon_\iota} \} $ and $ \{ \tilde{\Upsilon}_{\phi}^{\varepsilon_\iota} \}_{\epsilon} \leftarrow \{ \tilde{\Upsilon}_{\phi}^{\varepsilon_\iota} \} $\;
				}
				\If{$ \Im_{\epsilon} \leq \check{\Im}_o $}
				{
					\textbf{break}\;
				}
				$ o \leftarrow o+1 $\;
			}
			$ \mathcal{O}_{\varsigma} \vartriangleleft (f_{\epsilon}, o_{\epsilon}) $ and $ \bar{\zeta} \leftarrow \max{\{ \bar{\zeta}_{\epsilon}, \bar{\zeta} \}} $\;
			\For{$ \varepsilon_\iota $ in $ R_{\epsilon} $}
			{
				\If{$ \phi^{\varepsilon_\iota} $ in $ \Phi_{\varsigma}^{\varepsilon_\iota} $}
				{
					$ \ell_{\phi}^{\varepsilon_\iota} \leftarrow \ell_{\phi}^{\varepsilon_\iota} + l_{\epsilon} $\;
				}
				\Else
				{
					$ \Phi_{\varsigma}^{\varepsilon_\iota} \vartriangleleft \phi^{\varepsilon_\iota} $, $ \ell_{\phi}^{\varepsilon_\iota} \leftarrow l_{\epsilon} $, $ \Upsilon_{\varsigma}^{\varepsilon_\iota} \vartriangleleft \breve{\Upsilon}_{\phi}^{\varepsilon_\iota} \leftarrow \tilde{\Upsilon}_{\phi}^{\varepsilon_\iota} $\;
				}
			}
		}
		$ \mathcal{O} \vartriangleleft \mathcal{O}_{\varsigma} $\;
	}
\end{algorithm}

Before the scheduling process, the slot expressions of flow attributes are given with predetermined slot length $ T $ satisfying conditions $ T \mid \gcd{(P^n)} $, $ \mathring{j}_i \geq 2 $ and $ \mathring{d}_i -h_i > 0 $. They participate in the parallel HFG scheduling method shown as Algo. 2, which distributes flow partitions to their respective conquering units in line 1. For every partition, the corresponding flows are scheduled incrementally to get their optimal injection offset in lines 2-29, which reduces device nodes involved in each scheduling to their own routes. In line 2, flows are first sorted to determine the scheduling order as $ \Theta_{\varsigma} $ \cite{9155434}. During each flow $ f_{\epsilon} $ scheduling, its feasible offset $ o $ is traversed below the bound $ \bar{o}_{\epsilon} $ given in line 4. It is the smaller of flow period $ p_{\epsilon} $ and latency constraint (4). The optimal offset $ o_{\epsilon} $ is filtered with the greedy policy preferring a lower transmission real-time rate and slot occupancy rate expressed as
\begin{equation}
	 \min_{o_{\epsilon}}{ (1-\rho) \cdot \frac{1}{n_{\epsilon}} \sum_{f_i}{\frac{\vartheta_i}{\mathring{d}_i}} + \rho \cdot \frac{\bar{\zeta}_{\epsilon}}{\Lambda} }.
\end{equation}
It is designed to match the global optimization goal (8), where $ n_{\epsilon} $ counts the flows that have been scheduled in the current partition and $ \bar{\zeta}_{\epsilon} $ is the global maximum slot occupancy among all slots $ \mathcal{S} $ of the whole network in $ f_{\epsilon} $ scheduling. Further, the constant part in (10) during offset searching of each flow $ f_{\epsilon} $ is removed to simplify the local goal as
\begin{equation}
	 \min_{o_{\epsilon}}{ (1-\rho) \cdot \frac{o_{\epsilon}}{n_{\epsilon} \cdot \mathring{d}_{\epsilon}} + \rho \cdot \frac{\bar{\zeta}_{\epsilon}}{\Lambda} }.
\end{equation}
This optimal value is symbolized as $ \Im_{\epsilon} $ and obtained in lines 5-23. Specifically, with increasing offset $ o $ in lines 5\&23, its corresponding optimization value $ \Im_o $ is given in lines 6-17 and filtered in lines 18-22. In this key process, the critical cause of scheduling complexity is obtaining maximum slot occupancy $ \bar{\zeta}_o $ under each offset and recurring offset searching. We first discuss the resolutions for the former as follows.

Since retrieving the slot occupancy slot-by-slot within the bound $ C $ is computationally complex, the complexity is quite sensitive to flow periods that determine the bound. It exponentially expands as the multiplicity of periods. With the built hyper-flow graph (3) on each link $ \varepsilon_\iota \in \mathcal{E} $, we prove the equivalence between hyper-flow cliques and flow confluences on any slot under the support of following Lemma \cite{9714183}.
\newtheorem{lemma}{Lemma}
\begin{lemma}
	For $ n $ hyper-flow sequences formed as $ \aleph(k_i) = q_i +k_i \cdot \mathring{p}_i $, if any sequence pair $ \aleph(k_{\alpha}) $ and $ \aleph(k_{\beta}) $ satisfies that $ \gcd(\mathring{p}_{\alpha}, \mathring{p}_{\beta}) | (q_{\alpha} -q_{\beta}) $, there exist $ z \in \mathbb{Z} $ satisfying $ \mathring{p}_i | (z -q_i)$, $ i \in \{1,2,...,n\} $, and vice versa.
	\label{1}
\end{lemma}
It means that each hyper-flow clique indicates the flow confluences on a certain slot $ s^{\varepsilon_\iota} $, and the flow confluences on every slot $ s^{\varepsilon_\iota} $ are carried by a hyper-flow clique. That is, the occupancy of any slot is given by the combined weight of its corresponding clique. Thus, the maximum slot occupancy $ \bar{\zeta}_o $ is equal to the maximum combined weight among all hyper-flow cliques. The obtaining of maximum slot occupancy $ \bar{\zeta}_o $ is converted to a maximal clique enumeration problem (MCEP) as shown in Fig. \ref{fig1_3}(a). Compared to violently retrieve, the hyper-flow graph based pattern maps/embeds the redundant slot occupancy information into less maximal flow cliques $ \Upsilon^{\varepsilon_\iota} $. Essentially, it merges the identical or contained flow confluences on different slots.

\begin{figure}[!t]
	\centering
	\subfloat[\label{fig:a}]{
		\includegraphics[width=\columnwidth]{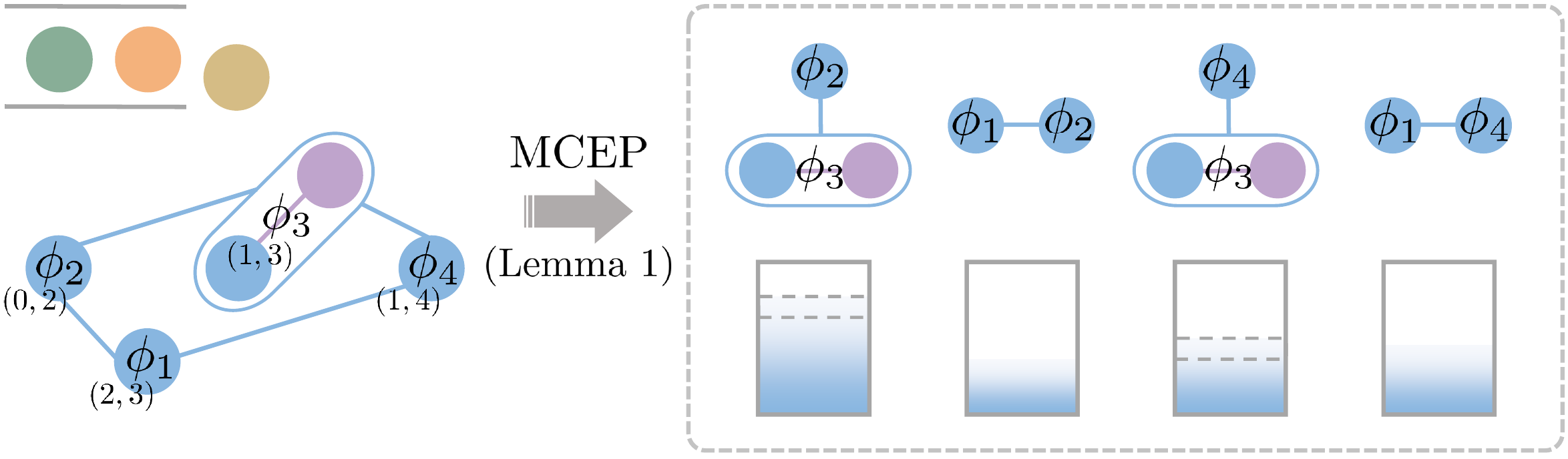}}
	
	\hspace{0.01mm}
	\subfloat[\label{fig:b}]{
		\includegraphics[width=\columnwidth]{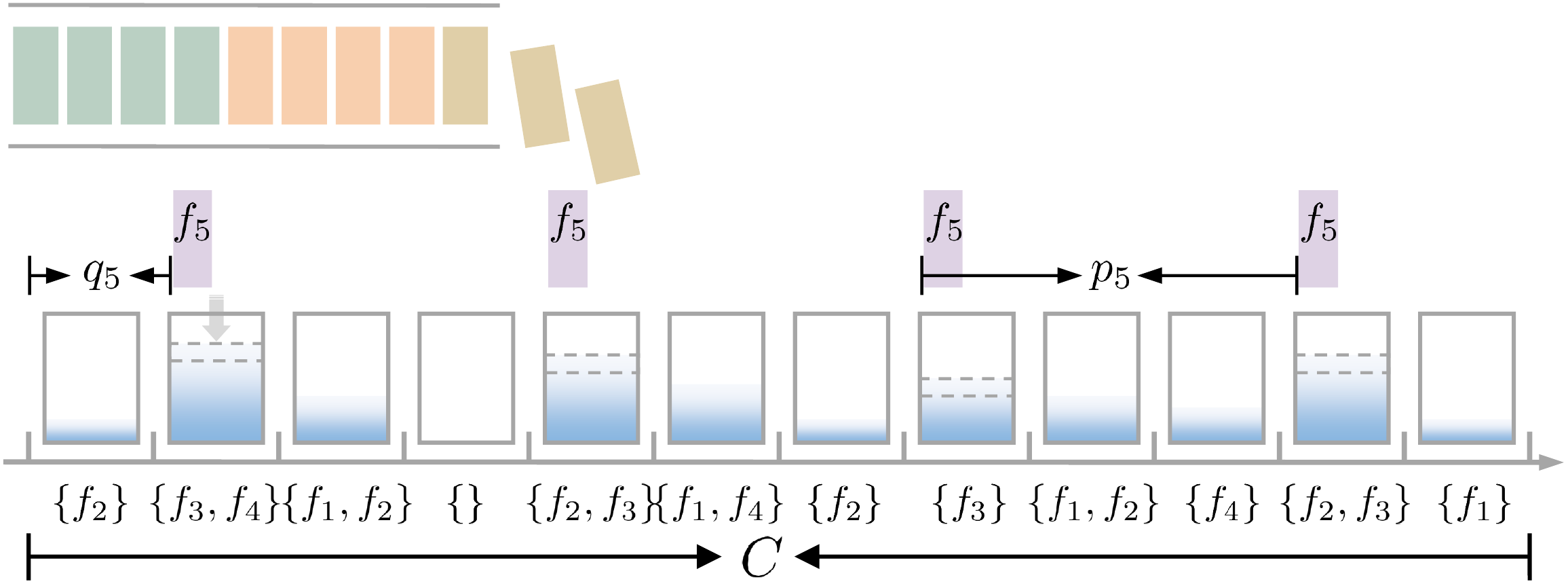}}
	
	\hspace{0.1mm}
	\subfloat[\label{fig:c}]{
		\includegraphics[width=\columnwidth]{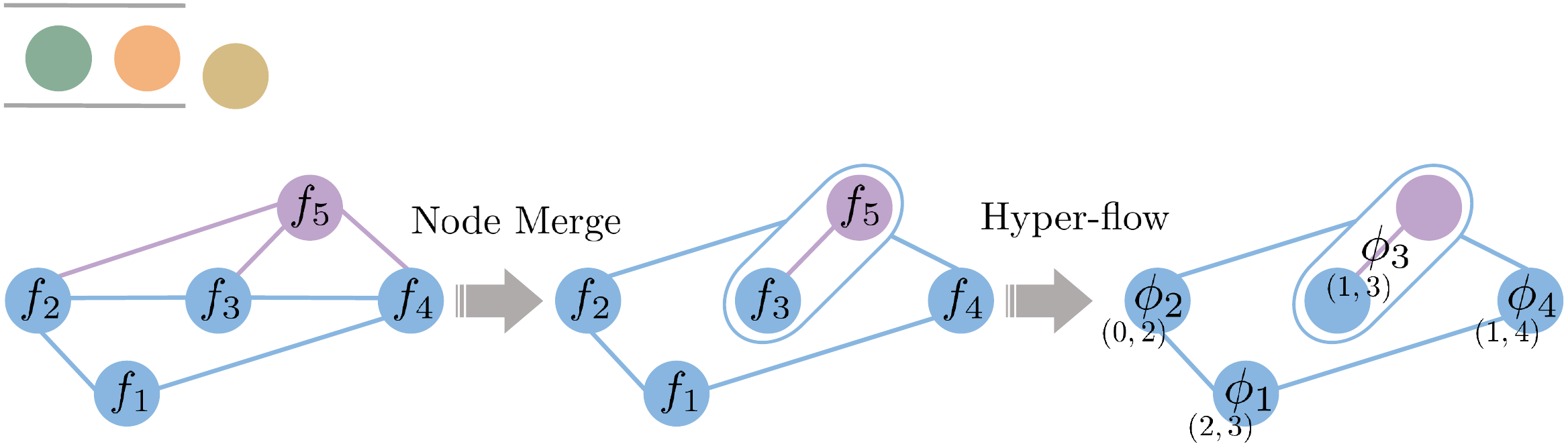}}
	
	\caption{(a) Hyper-flow graph based pattern, (b) Frame-based pattern, (c) Flow graph-based pattern}
	\label{fig1_3} 
\end{figure}

In this way, we design a dynamic graph and clique update method to obtain the maximum slot occupancy $ \bar{\zeta}_o $ under offset $ o $ of flow $ f_{\epsilon} $ for incremental scheduling. As shown in lines 6-16 of Algo. 2, for each link $ \varepsilon_\iota $ along route $ R_{\epsilon} $, the feature tuple of flow $ f_{\epsilon} $ is extracted in line 7 and checked its existence in the hyper-flow nodes $ \Phi_{\varsigma}^{\varepsilon_\iota} $. It covers hyper-flows that have been scheduled on link $ \varepsilon_\iota $. If it exists, the weight $ \ell_{\phi}^{\varepsilon_\iota} $ of corresponding $ \phi^{\varepsilon_\iota} $ increases by $ l_{\epsilon} $ (without real update in the pseudo-code during the offset search), and so do all combined weights $ \zeta_{\varpi}^{\varepsilon_\iota} $ of maximal cliques $ \varpi_{\kappa}^{\varepsilon_\iota} \in \Upsilon_{\phi}^{\varepsilon_\iota} $ that contains $ \phi^{\varepsilon_\iota} $ among the maximal hyper-flow clique set $ \Upsilon_{\varsigma}^{\varepsilon_\iota} $ in lines 9-10. They imply the potential maximum occupancy of the slots containing $ f_{\epsilon} $. If not, the tuple extends the graph as a new hyper-flow $ \phi^{\varepsilon_\iota} $. Its neighbors $ \varkappa_{\phi}^{\varepsilon_\iota} $ are given by (3) in line 12, where the neighbors mean its connected nodes by (3) in the graph. The potential maximum slot occupancy containing $ f_{\epsilon} $ is obtained based on the following theorem.
\newtheorem{theorem}{Theorem}
\begin{theorem}
In the graph with maximal cliques $ \varpi_{\kappa} \in \Upsilon $, if a new node $ \phi $ joins with the neighbors $ \varkappa_{\phi} $, the new maximal cliques $ \breve{\varpi}_{\kappa} \in \breve{\Upsilon} $ compared to $ \Upsilon $ are included in the clique set
\begin{equation}
	\tilde{\Upsilon}_{\phi} = \{ \tilde{\varpi}_{\kappa} \mid \tilde{\varpi}_{\kappa} = (\varpi_{\kappa} \cap \varkappa_{\phi}) \cup \{ \phi \},\ \varpi_{\kappa} \in \Upsilon_{\varkappa} \},
\end{equation}
where $ \Upsilon_{\varkappa} $ is $ \{ \varpi_{\kappa} \in \Upsilon \mid (\varpi_{\kappa} \cap \varkappa_{\phi}) \neq \varnothing \} $ as the subset of $ \Upsilon $. Specially, it is modified to set $ \{ \varnothing \} $ when empty. Furthermore, for any clique $ \tilde{\varpi}_{\kappa} \in \tilde{\Upsilon}_{\phi} - \breve{\Upsilon} $, there exits a new maximal clique $ \breve{\varpi}_{\kappa} \in \breve{\Upsilon} $ satisfying that $ \tilde{\varpi}_{\kappa} \subseteq \breve{\varpi}_{\kappa} $.
\label{1}
\end{theorem}

\begin{proof}
It is deduced and verified step by step as follows.

Considering that $ (\varpi_{\kappa} \cap \varkappa_{\phi}) \subseteq \varkappa_{\phi} $, any element $ \phi_j \in (\varpi_{\kappa} \cap \varkappa_{\phi}) $ is connected with $ \phi $. Moreover, since the subset $ (\varpi_{\kappa} \cap \varkappa_{\phi}) $ of clique $ \varpi_{\kappa} \in \Upsilon $ must be a clique, we prove that $ (\varpi_{\kappa} \cap \varkappa_{\phi}) \cup \{ \phi \} $ is a clique and $ \tilde{\Upsilon}_{\phi} $ is the clique set.

For any new maximal clique $ \breve{\varpi}_{\kappa} \in \breve{\Upsilon} $, it is easy to deduced that $ \phi \in \breve{\varpi}_{\kappa} $. Otherwise, $ \breve{\varpi}_{\kappa} $ is constituted by the original nodes of the graph and $ \breve{\varpi}_{\kappa} \in \Upsilon $, which violates the definition of $ \breve{\varpi}_{\kappa} $. Considering the full connection feature between nodes of the clique, the remaining elements of $ \breve{\varpi}_{\kappa} $ except $ \phi $ only exist in its neighbors, that is, $ \breve{\varpi}_{\kappa} \subseteq (\varkappa_{\phi} \cup \{ \phi \}) $. In addition, the fully connected remaining elements must be included in a certain maximal clique $ \varpi_{\kappa} \in \Upsilon $, that is, $ \breve{\varpi}_{\kappa} \subseteq (\varpi_{\kappa} \cup \{ \phi \}) $. Hence, there exists $ \varpi_{\kappa} \in \Upsilon $ satisfying that $ \breve{\varpi}_{\kappa} \subseteq (\varpi_{\kappa} \cup \{ \phi \}) \cap (\varkappa_{\phi} \cup \{ \phi \}) $, that is, $ \breve{\varpi}_{\kappa} \subseteq (\varpi_{\kappa} \cap \varkappa_{\phi}) \cup \{ \phi \} $. At this time, the maximal clique $ \varpi_{\kappa} \in \Upsilon_{\varkappa} $ also exists if the above existing $ \varpi_{\kappa} $ satisfies the condition $ (\varpi_{\kappa} \cap \varkappa_{\phi}) \neq \varnothing $. Otherwise, all these existing $ \varpi_{\kappa} $ satisfies that $ (\varpi_{\kappa} \cap \varkappa_{\phi}) = \varnothing $ and $ \breve{\varpi}_{\kappa} \subseteq \{ \phi \} $, which means the neighbors $ \varkappa_{\phi} $ are the empty set $ \varnothing $. Thus, $ \Upsilon_{\varkappa} $ is the set $ \{ \varnothing \} $. There exists $ \varpi_{\kappa} = \varnothing $ satisfying that $ \breve{\varpi}_{\kappa} \subseteq (\varpi_{\kappa} \cap \varkappa_{\phi}) \cup \{ \phi \} $. Since the maximal clique cannot be the proper subset of any other clique, for any $ \breve{\varpi}_{\kappa} \in \breve{\Upsilon} $, there exists $ \varpi_{\kappa} \in \Upsilon_{\varkappa} $ satisfying the condition $ \breve{\varpi}_{\kappa} = (\varpi_{\kappa} \cap \varkappa_{\phi}) \cup \{ \phi \} $. Hence, $ \breve{\Upsilon} \subseteq \tilde{\Upsilon}_{\phi} $ holds.

Similarly, for any clique $ \tilde{\varpi}_{\kappa} \in \tilde{\Upsilon}_{\phi} - \breve{\Upsilon} $, since node $ \phi \in \tilde{\varpi}_{\kappa} $, we have that $ \phi \notin \varpi_{\kappa} $, $ \forall \varpi_{\kappa} \in \Upsilon $. Thus, there must exists the new maximal clique $ \breve{\varpi}_{\kappa} \in \breve{\Upsilon} $ satisfying the condition $ \tilde{\varpi}_{\kappa} \subseteq \breve{\varpi}_{\kappa} $ . Otherwise, $ \tilde{\varpi}_{\kappa} $ is verified as a new maximal clique, which conflicts with $ \tilde{\varpi}_{\kappa} \in \tilde{\Upsilon}_{\phi} - \breve{\Upsilon} $.
\end{proof}

According to this theorem, the maximum occupancy of the slots containing $ f_{\epsilon} $ is implied in the combined weights $ \zeta_{\varpi}^{\varepsilon_\iota} $ of cliques $ \tilde{\varpi}_{\kappa}^{\varepsilon_\iota} \in \tilde{\Upsilon}_{\phi}^{\varepsilon_\iota} $ given by (12). These values are obtained in lines 13-15. Based on this, the maximum slot occupancy $ \bar{\zeta}_o $ under each offset is measured by comparing $ \max{\{ \zeta_{\varpi}^{\varepsilon_\iota} \}} $ and the previous maximum occupancy $ \bar{\zeta} $ before flow $ f_{\epsilon} $ scheduling in line 16. Then, the optimization value $ \Im_o $ for offset $ o $ during flow $ f_{\epsilon} $ scheduling is given by (11) in line 17, which works as the filtering basis for the optimal value $ \Im_{\epsilon} $ and offset $ o_{\epsilon} $ in lines 18-20. Under each offset $ o $, the computational complexity to get $ \bar{\zeta}_o $  is less than $ \min{\{ (m-\rho) \cdot 3^{\rho/3}, \lcm(\mathring{P}^n) / \mathring{p}_{\epsilon} \}} $ on each link, where $ m $ is the neighbor number of flow $ f_{\epsilon} $ in the flow graph. $ \rho $ are the degeneracy of the neighbor built sub-graphs.

Compared to the above hyper-flow graph based pattern, the traditional frame-based scheduling \cite{9155434,9407828,9893358,10101832} builds a slot occupancy list covering all slots $ \mathcal{S} $ of the whole network as shown in Fig. \ref{fig1_3}(b). Due to the uncertainty of which slots have current maximum occupancy, it puts all periodic frames of flow $ f_{\epsilon} $ into their specific slots $ s^{\varepsilon_\iota} \in \mathcal{S} $ within the bound $ C $ along every link $ \varepsilon_\iota \in R_{\epsilon} $ during each offset search. The occupancy $ \zeta_s^{\varepsilon_\iota} $ of these slots is measured one-by-one. This flow period-sensitive pattern owns the computational complexity as constant $ \lcm(\mathring{P}^n) / \mathring{p}_{\epsilon} $ on each link. For the large-scale complex flows, it almost exponentially expands and severely impacts scheduling efficiency.

In our previous work \cite{9714183}, a flow graph-based scheduling pattern is explored to obtain the key scheduling information, that is, maximum slot occupancy $ \bar{\zeta}_o $. Distinguished from the hyper-flow nodes, this graph is built by treating each flow as a node and connects edges with a similar matrix $ \Psi^{\varepsilon_\iota} $
\begin{equation}
	\Psi_{f_{\alpha}, f_{\beta}}^{\varepsilon_\iota} =
	\begin{cases}
		1 & \text{if\ } \gcd(\mathring{p}_{\alpha}, \mathring{p}_{\beta}) | (q_{\alpha}^{\varepsilon_\iota} -q_{\beta}^{\varepsilon_\iota}). \\
		0 & \text{otherwise.}
	\end{cases}
\end{equation}
As shown in Fig. \ref{fig1_3}(c), it is easy to prove that the hyper-flow graph can be shrunk from the flow graph with the following node merging principles.
\begin{enumerate}
	\item The merged flow nodes are interconnected in pairs with edges by (13);
	\item The merged flow nodes have the same neighbors $ \varkappa_{\phi} $ except for each other;
	\item After merging, the new node inherits the weight as the sum of the merged flow nodes.
\end{enumerate}
Thus, these two patterns share the same number and combined weights of maximal cliques. Their computational complexity is similar, but the hyper-flow graph pattern has smaller $ m $ and $ \rho $. Since the complexity of MCEP rises exponentially along with the node scale \cite{2543630}, the smaller hyper-flow graph facilitates the large-scale flow scheduling. It also curbs the exponential trend of MECP due to its slower growth. To sum up, the hyper-flow graph based scheduling pattern suppresses the sensitivity of scheduling complexity to flow attributes, while weakening the impact of flow scale on it.

Since the recurring offset searching accumulates the above complexity, we further improve offset filtering efficiency with an early-break strategy. It provides a breaking condition for offset traversal before reaching bound $ \bar{o}_{\epsilon} $. Specifically, the optimization values under offsets behind current $ o $ are compared with current optimal value $ \Im_{\epsilon} $ in lines 21-22. Compared to this optimal $ \Im_{\epsilon} $, if all values $ \Im_{\tilde{o}} $ under offset $ \tilde{o} > o $ satisfy that
\begin{equation}
	 \Im_{\epsilon} \leq \inf_{\tilde{o} > o}{\Im_{\tilde{o}}} = \inf_{\tilde{o} > o}{ \{ (1-\rho) \cdot \dfrac{\tilde{o}}{n_{\epsilon} \cdot \mathring{d}_{\epsilon}} + \rho \cdot \frac{\bar{\zeta}_{\tilde{o}}}{\Lambda} \} }.
\end{equation}
It means that the current value $ \Im_{\epsilon} $ and its offset $ o_{\epsilon} $ is optimal for flow $ f_{\epsilon} $ scheduling. Thus, we break the offset search from its traversal. Based on this, we deflate (14) as
\begin{equation}
	 \Im_{\epsilon} \leq \inf_{\tilde{o} > o}{ \{ (1-\rho) \cdot \dfrac{\tilde{o}}{n_{\epsilon} \cdot \mathring{d}_{\epsilon}} \} } + \inf_{\tilde{o} > o}{ \{  \rho \cdot \frac{\bar{\zeta}_{\tilde{o}}}{\Lambda} \} }.
\end{equation}
It is easy to prove that (14) must hold if the condition (15) is satisfied. Further, we have that
\begin{equation}
	\begin{aligned}
	& \inf_{\tilde{o} > o}{ \{ (1-\rho) \cdot \dfrac{\tilde{o}}{n_{\epsilon} \cdot \mathring{d}_{\epsilon}} \} } = (1-\rho) \cdot \dfrac{o+1}{n_{\epsilon} \cdot \mathring{d}_{\epsilon}}, \\
	& \inf_{\tilde{o} > o}{ \{  \rho \cdot \frac{\bar{\zeta}_{\tilde{o}}}{\Lambda} \} } = \rho \cdot \frac{\max{ \{ \bar{\zeta}, l_{\epsilon} \} }}{\Lambda},
	\end{aligned}
	\nonumber
\end{equation}
which builds the early-break threshold under offset $ o $ as
\begin{equation}
	 \check{\Im}_o = (1-\rho) \cdot \dfrac{o+1}{n_{\epsilon} \cdot \mathring{d}_{\epsilon}} + \rho \cdot \frac{\max{ \{ \bar{\zeta}, l_{\epsilon} \} }}{\Lambda}.
\end{equation}
If the condition $ \Im_{\epsilon} \leq \check{\Im}_o $ is satisfied, the offset traversal ends immediately in line 22. Unusually, when weight factor $ \rho $ is 1, the early-break condition is converted into $ \bar{\zeta}_{\epsilon} \leq \max{ \{ \bar{\zeta}, l_{\epsilon} \} } $ introduced in work \cite{9407828}. 

After the flow $ f_{\epsilon} $ scheduling by offset search, the optimal offset $ o_{\epsilon} $ is stored and the global maximum slot occupancy $ \bar{\zeta} $ is updated in line 24, where the symbol $ \vartriangleleft $ represents the store operation. Then, the hyper-flow graph and its maximal cliques on each link $ \varepsilon_\iota \in R_{\epsilon} $ are updated in lines 25-29. Specifically, if there exists the hyper-flow $ \phi^{\varepsilon_\iota} $ with the same feature tuple as optimal offset $ o_{\epsilon} $, this hyper-flow adds frame length $ l_{\epsilon} $ to its weight $ \ell_{\phi}^{\varepsilon_\iota} $ in line 27. Otherwise, the tuple is extended as a new hyper-flow with weight $ l_{\epsilon} $, and stored in line 29 with its corresponding potential maximal cliques $ \tilde{\Upsilon}_{\phi}^{\varepsilon_\iota} $. Due to Theorem 1, they need to be reduced to maximal hyper-flow cliques $ \breve{\Upsilon}_{\phi}^{\varepsilon_\iota} $ and then stored. In the incremental way, every flow partition is scheduled and their flow injection offsets $ \mathcal{O}_{\varsigma} $ are aggregated as a preliminary solution $ \mathcal{O} $ in line 30.

\subsection{Link-parallel Flow Synthesis Phase}
Due to the network sharing of sub-problems in the above partition scheduling, their common slot occupancy affects load balancing and causes potential overflow. Hence, the global slot occupancy information, that is, maximal hyper-flow cliques, needs to be synthesized from these sub-problems. The complete hyper-flow graphs are built link-by-link that accumulates the device scale induced scheduling complexity. To deal with this, we adopt a link-parallel structure that spreads the hyper-flow graph synthesis for each link $ \varepsilon_\iota \in \mathcal{E} $ over multiple computing units as shown in Fig. \ref{linkparallel}.

\begin{figure}[!h]
	\centerline{\includegraphics[width=6cm]{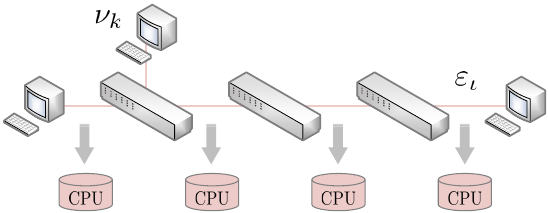}}
	\caption{Link-parallel synthesis structure}
	\label{linkparallel}
\end{figure}

In the second phase, each scheduling sub-problem identifies its hyper-flows on each link. Hence, the graph and maximal cliques synthesis could be separated across links. The pseudo-code is shown as Algo. 3. Available computing units are first allocated to every link containing target flows in line 1. For any link $ \varepsilon_\iota \in \mathcal{E} $, the hyper-flows $ \Phi_{\varsigma}^{\varepsilon_\iota} $ and their corresponding weights $ \{ \ell_{\phi}^{\varepsilon_\iota} \}_{\varsigma}^{\varepsilon_\iota} $ from sub-problems are synthesized by (3) in lines 2-3. It forms a whole hyper-flow graph with nodes $ \Phi^{\varepsilon_\iota} $ and weights $ \mathcal{L}^{\varepsilon_\iota} $, where the same nodes are merged and new edges are connected. Meanwhile, the neighbors are accumulated as $ \mathcal{N}^{\varepsilon_\iota} $ for each hyper-flow in the synthesized graph in line 4. Based on this, line 5 adopts a Bron-Kerbosch Degeneracy method \cite{2543629} to enumerate the global maximal cliques $ \Upsilon^{\varepsilon_\iota} $ on each link. Also, their combined weights are counted and aggregated in lines 6-8. After this, the global slot occupancy information is given in line 8 and stored as $ \Upsilon $.

\begin{algorithm}[!h]
	\caption{Parallel Bron-Kerbosch Method}
	\KwIn{Hyper-flows: $ \{ \Phi_{\varsigma}^{\varepsilon_\iota} \}_{\varsigma} \in \Phi $ \\
		Hyper-flow weights: $ \{ \ell_{\phi}^{\varepsilon_\iota} \}_{\varsigma} \in \mathcal{L} $
	}
	\KwOut{Global maximal cliques: $ \{ \varpi_{\kappa}^{\varepsilon_\iota} \}^{\varepsilon_\iota} \in \Upsilon $
	}
	\for{$ \varepsilon_\iota $ in $ \mathcal{E} $}
	{
		$ \Phi^{\varepsilon_\iota} \leftarrow $ merge the same $ \phi^{\varepsilon_\iota} $ from all partitions\;
		$ \mathcal{L}^{\varepsilon_\iota} \leftarrow $ sum up weights $ \ell_{\phi}^{\varepsilon_\iota} $ for the same $ \phi^{\varepsilon_\iota} $\;
		$ \mathcal{N}^{\varepsilon_\iota} \leftarrow $ get neighbors $ \varkappa_{\phi}^{\varepsilon_\iota} $ for each $ \phi^{\varepsilon_\iota} \in \Phi^{\varepsilon_\iota} $\;
		$ \Upsilon^{\varepsilon_\iota} \leftarrow $ Bron-Kerbosch Degeneracy($ \Phi^{\varepsilon_\iota}, \mathcal{N}^{\varepsilon_\iota} $)\;
		\For{$ \varpi_{\kappa}^{\varepsilon_\iota} $ in $ \Upsilon^{\varepsilon_\iota} $}
		{
			$ \zeta_{\varpi}^{\varepsilon_\iota} \leftarrow \sum_{\phi \in \varpi_{\kappa}^{\varepsilon_\iota}} \ell_{\phi}^{\varepsilon_\iota} $\;
		}
		$ \Upsilon \vartriangleleft \Upsilon^{\varepsilon_\iota} $\;
	}
\end{algorithm}

\subsection{Precise Flow Fine-tuning Phase}
In this phase, global hyper-flow graphs and their maximal cliques are involved in overflow avoidance to improve schedulability and load-balancing. The potential overflow slots are desired to precisely fine-tune the flows inside them, while the optimized QoS performances are maintained. To handle this, the conflict clique is first defined and proved its superiority for flow fine-tuning. It is an extension of the above hyper-flow graph based methodology and drives the design of an efficient CCR re-scheduling method shown in Fig. \ref{CCR}, including the following overflow checking, conflict clique detecting, sub-flow backtracking and sub-flow re-scheduling.

\begin{figure*}[!t]
	\centerline{\includegraphics[width=18.2cm]{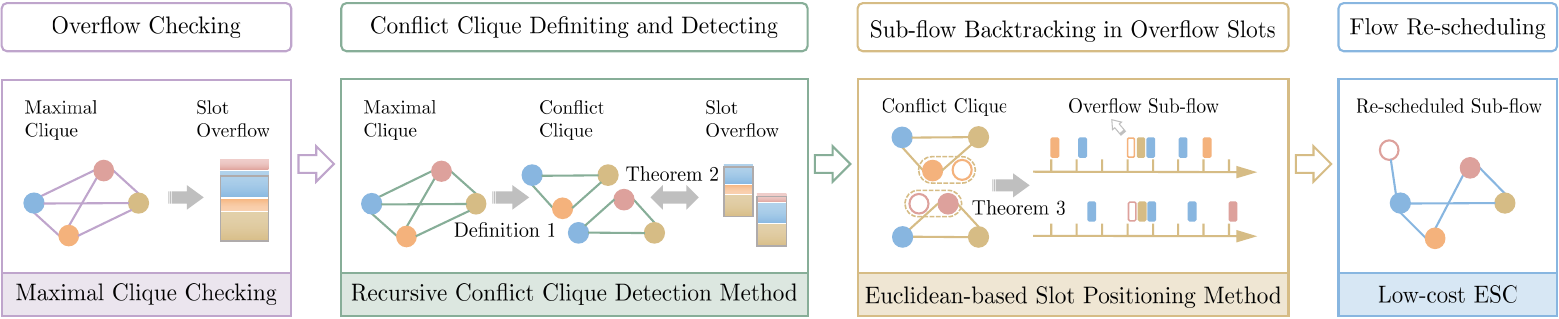}} 
	\caption{Processes of CCR re-scheduling method}
	\label{CCR}
\end{figure*}

Limited by available slot capacity $ \Lambda $ in constraint (7), the global slot occupancy needs to be checked for overflow. According to the maximal hyper-flow cliques from the synthesis phase, we give the overflow maximal cliques $ \hat{\varpi}_{\kappa}^{\varepsilon_\iota} \in \hat{\Upsilon}^{\varepsilon_\iota} $ on each link $ \varepsilon_\iota \in \mathcal{E} $ as
\begin{equation}
	 \hat{\Upsilon}^{\varepsilon_\iota} = \{ \varpi_{\kappa}^{\varepsilon_\iota} \mid \varpi_{\kappa}^{\varepsilon_\iota} \in \Upsilon^{\varepsilon_\iota} \text{ and } \zeta_{\varpi}^{\varepsilon_\iota} > \Lambda \},
\end{equation}
which are aggregated as $ \hat{\Upsilon} $. With further analysis of Lemma 1, we deduce that the flow confluences on overflow slots $ \hat{\mathcal{S}} $ must be carried by the sub-cliques in these cliques $ \hat{\Upsilon} $. To avoid slot overflow, a naive way is to backtrack and re-schedule the flows $ \dot{f}_i \in \dot{\mathcal{F}} $ satisfying that
\begin{equation}
	 \forall \hat{\varpi}_{\kappa}^{\varepsilon_\iota} \in \hat{\Upsilon},\ \text{the combined weight of } \{ \hat{\varpi}_{\kappa}^{\varepsilon_\iota} - \dot{\mathcal{F}} \} \leq \Lambda.
	 \nonumber
\end{equation}
With different flow offsets re-filtered by Algo. 2, the hyper-flow graphs are modified to eliminate the overflow maximal cliques $ \hat{\Upsilon} $, that is, the overflow slots. Since these flows occupy not only overflow slots but also others, the whole flow re-scheduling impacts current slot allocation and QoS performances more. To focus on overflow avoidance while maintaining the QoS, we aim to precisely position the overflow slots $ \hat{s}^{\varepsilon_\iota} \in \hat{\mathcal{S}} $ and fine-tune the flow portions inside them. It breaks the assumption that all frames of each flow share the same delay and makes the preliminary solution be rectified as slightly as possible. Between the overflow slots $ \hat{\mathcal{S}} $, flow portions and given overflow maximal cliques $ \hat{\Upsilon} $, we introduce the conflict clique $ \dot{\varpi}_{\kappa}^{\varepsilon_\iota} $ to build their correlation. This concept is first given by us and defined as
\newtheorem{definition}{Definition}
\begin{definition}
In a hyper-flow graph with weighted nodes, the conflict clique $ \dot{\varpi}_{\kappa}^{\varepsilon_\iota} $ is a clique with the features as
\begin{enumerate}
	\item Its combined weight exceeds slot capacity $ \Lambda $;
	\item The weight of its any proper sub-clique is below $ \Lambda $.
\end{enumerate}
\label{1}
\end{definition}
Under such conditions, these conflict cliques $ \dot{\Upsilon} $ have the advance in re-scheduling shown as the following theorem.
\begin{theorem}
For all conflict cliques $ \dot{\varpi}_{\kappa}^{\varepsilon_\iota} \in \dot{\Upsilon} $ under the above definition, they satisfy that
\begin{enumerate}
	\item The confluence slots $ \dot{s}^{\varepsilon_\iota} \in \dot{\mathcal{S}} $ of flows in conflict cliques are equivalent to overflow slots $ \hat{\mathcal{S}} $;
	\item The slot overflow would be avoided by backtracking and re-scheduling the flow portions $\dot{\mathcal{F}} $ given as 
	\begin{equation}
	\{ f_i^{\dot{\varpi}} \mid \text{ flows } \{ f_i \} \text{ of any one node in } \dot{\varpi}_{\kappa}^{\varepsilon_\iota},\ \dot{\varpi}_{\kappa}^{\varepsilon_\iota} \in \dot{\Upsilon} \},
	 \nonumber
	\end{equation}
	where $ f_i^{\dot{\varpi}} $ represents the flow $ f_i $ portion that occupies the confluence slots $ \dot{s}^{\varepsilon_\iota} $ of clique $ \dot{\varpi}_{\kappa}^{\varepsilon_\iota} $.
\end{enumerate}
\label{2}
\end{theorem}
\begin{proof}
The propositions are verified step by step as follows.
	
For any conflict clique $ \dot{\varpi}_{\kappa}^{\varepsilon_\iota} \in \dot{\Upsilon} $, it is easy to get that the corresponding flow confluence slots overflow due to feature 1). Conversely, for any overflow slot $ \hat{s}^{\varepsilon_\iota} \in \hat{\mathcal{S}} $ and its corresponding flow confluences as the clique $ \varpi $ by Lemma 1, there must exist a sub-clique satisfying all features of conflict cliques. Moreover, the flow confluence slots of this sub-clique must cover the confluence slots of $ \varpi $, that is, the overflow slot $ \hat{s}^{\varepsilon_\iota} $. Thus, proposition 1) holds.

During the backtracking for re-scheduling following proposition 2), each conflict clique $ \dot{\varpi}_{\kappa}^{\varepsilon_\iota} \in \dot{\Upsilon} $ takes out one node with the flows set $ \{ f_i \} $. It means that their flow portions $ f_i^{\dot{\varpi}} $ in the confluence slots $ \dot{s}^{\varepsilon_\iota} $ of conflict clique $ \dot{\varpi}_{\kappa}^{\varepsilon_\iota} $ would be re-scheduled outside these slots. In this way, there is no conflict clique satisfying the features 1)\&2) among the network. For any overflow slot, its flow confluences would be below the slot capacity. Otherwise, the overflow flow confluences imply the existence of conflict cliques by proposition 1), which violates the above inferences. Thus, proposition 2) holds.
\end{proof}

Based on this theorem, we design the following methods to precisely position overflow slots and fine-tune their flow portions, including recursive conflict clique detection, Euclidean-based slot positioning and low-cost sub-flow re-scheduling.

\subsubsection{Recursive Conflict Clique Detection}
For the Theorem 2-based re-scheduling, all conflict cliques $ \dot{\Upsilon} $ need to be detected from the overflow maximal cliques $ \hat{\Upsilon} $. The detection pseudo-code is shown as Algo. 4, where a specific recursive process is illustrated independently. In lines 1-3, each overflow maximal clique $ \hat{\varpi}_{\kappa}^{\varepsilon_\iota} \in \hat{\Upsilon}^{\varepsilon_\iota} $ on every link $ {\varepsilon_\iota} \in \mathcal{E} $ is sorted its hyper-flows in ascending order of their weights. The ordered queue $ \Theta_{\kappa}^{\varepsilon_\iota} $ works for the ConflictCliqueDetection process along with the combined weight $ \hat{\zeta}_{\kappa}^{\varepsilon_\iota} $ of $ \hat{\varpi}_{\kappa}^{\varepsilon_\iota} $, the slot capacity $ \Lambda $ and temporary variable $ \omega $ initialized to $ \varnothing $.
\begin{equation}
	\begin{aligned}
		& \textbf{proc } \text{ConflictCliqueDetection} (\theta, \zeta, \Lambda, \omega) \\
		& \ \text{1:} \ \dot{\varpi} \leftarrow \omega \cup \theta \text{ and get the first element } \phi \in \dot{\varpi} \text{;} \\
		& \ \text{2:} \ \textbf{if } \zeta -\ell_{\phi} > \Lambda \ \textbf{then} \\
		& \ \text{3:} \ \quad \text{report } \dot{\varpi} \text{ as a conflict clique and }\textbf{return;}\\
		& \ \text{4:} \ \textbf{for } \phi_i \text{ in } \theta \textbf{ do} \\
		& \ \text{5:} \ \quad \textbf{if } \zeta -\ell_{\phi_i} > \Lambda \ \textbf{then} \\
		& \ \text{6:} \ \quad \quad \tilde{\theta} \leftarrow \{ \phi_{i+1}, \phi_{i+2}, \cdots \in \theta \} \text{ and } \tilde{\omega} \leftarrow \{ \cdots, \phi_i \in \theta \} \\
		& \ \text{7:} \ \quad \quad \text{ConflictCliqueDetection} (\tilde{\theta}, \zeta-\ell_{\phi_i}, \Lambda, \{ \omega, \tilde{\omega} \})
	\end{aligned}
	\nonumber
\end{equation}

\begin{algorithm}[!t]
	\caption{Recursive Conflict Clique Detection}
	\KwIn{Overflow maximal cliques: $ \{ \hat{\varpi}_{\kappa}^{\varepsilon_\iota} \}^{\varepsilon_\iota} \in \hat{\Upsilon} $ \\
		Available slot capacity: $ \Lambda $
	}
	\KwOut{Conflict cliques: $ \{ \dot{\varpi}_{\kappa}^{\varepsilon_\iota} \}^{\varepsilon_\iota} \in \dot{\Upsilon} $
	}
	\For{$ \varepsilon_\iota $ in $ \mathcal{E} $}
	{
		\For{$ \hat{\varpi}_{\kappa}^{\varepsilon_\iota} $ in $ \hat{\Upsilon}^{\varepsilon_\iota} $}
		{
			$ \Theta_{\kappa}^{\varepsilon_\iota} \leftarrow \{ \phi_1^{\varepsilon_\iota}, \phi_2^{\varepsilon_\iota}, \cdots \in \hat{\varpi}_{\kappa}^{\varepsilon_\iota} \ | \ \ell_{\phi_1}^{\varepsilon_\iota} < \ell_{\phi_2}^{\varepsilon_\iota} < \cdots \} $\;
			$ \dot{\Upsilon}^{\varepsilon_\iota} \vartriangleleft $ ConflictCliqueDetection($ \Theta_{\kappa}^{\varepsilon_\iota}, \Lambda, \hat{\zeta}_{\kappa}^{\varepsilon_\iota}, \varnothing $)\;
		}
		$ \dot{\Upsilon} \vartriangleleft \dot{\Upsilon}^{\varepsilon_\iota} $\;
	}
\end{algorithm}

\begin{figure}[!h]
	\centerline{\includegraphics[width=\columnwidth]{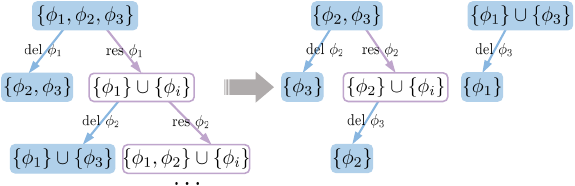}}
	\caption{Recursive conflict clique detection process}
	\label{fig5}
\end{figure}

During the process, conflict cliques implied in every maximal clique $ \hat{\varpi}_{\kappa}^{\varepsilon_\iota} $ are detected with the iteration of cliques. In each iteration, the given clique $ \dot{\varpi} $ satisfying feature 1) of the conflict clique is composed of the reserved part $ \omega $ and flexible part $ \theta $ in line 1 and checked if the feature 2) holds in line 2. Specifically, feature 2) is checked by comparing slot capacity $ \Lambda $ and the sub-clique removing the hyper-flow node with the minimum weight, that is, the first element $ \phi \in \dot{\varpi} $. The checking passed clique is verified as a conflict clique and gathered into set $ \dot{\Upsilon} $. At this time, its corresponding iteration ends in line 3. If feature 2) is not satisfied, it means that there exist conflict cliques among the proper sub-cliques of $ \dot{\varpi} $. Hence, by popping one hyper-flow $ \phi_i $ from the flexible part $ \theta $ of clique $ \dot{\varpi} $, the sub-cliques of $ \dot{\varpi} $ are filtered by feature 1) and then traversed for the next round of iterations in lines 4-7. In this way, the clique iterations start from maximal clique $ \hat{\varpi}_{\kappa}^{\varepsilon_\iota} $ and continue round by round with the scale of checked cliques decreasing. The iterations end until all conflict cliques are detected. In the same round, they shared the same clique scale. It is worth noting that in order to avoid iterating and checking the duplicate cliques, we distinguish the flexible part $ \theta $ and reserved part $ \omega $ for each detection process, which varies with rounds of iterations in lines 6-7. The former is traversed for clique iteration rather than the latter. It makes the sub-cliques unique like Fig. \ref{fig5} in each round of iterations.

In this way, all conflict cliques $ \dot{\Upsilon} $ are obtained from every overflow maximal clique $ \hat{\varpi}_{\kappa}^{\varepsilon_\iota} \in \hat{\Upsilon} $ in line 4 of Algo. 4. The conflict cliques on links $ \varepsilon_\iota \in \mathcal{E} $ are stored as $ \dot{\Upsilon}^{\varepsilon_\iota} $ and then aggregated as set $ \dot{\Upsilon} $ in line 5. On each link, the same conflict cliques $ \dot{\varpi}_{\kappa}^{\varepsilon_\iota} $ may be implied in multiple overflow maximal cliques $ \hat{\varpi}_{\kappa}^{\varepsilon_\iota} $. They are de-duplicated into one clique.

\subsubsection{Euclidean-based Slot Positioning}
Based on Theorem 2 and given conflict cliques $ \dot{\Upsilon} $, we aim to position the confluence slots $ \dot{s}^{\varepsilon_\iota} $ of every conflict clique $ \dot{\varpi}_{\kappa}^{\varepsilon_\iota} $ and backtracking their specific flow portions $ \dot{\mathcal{F}} $. Hence, a precise clique-aware slot positioning method is designed under the guidance of the following theorem.
\begin{theorem}
For $ n $ sequences as $ \aleph(k_i) = q_i +k_i \cdot \mathring{p}_i $, where $i \in \{1,2,...,n\}$, there are two equivalent propositions as
\begin{enumerate}
	\item For any sequence pair $ \aleph(k_{\alpha}) $ and $ \aleph(k_{\beta}) $, it satisfies that $ \gcd(\mathring{p}_{\alpha}, \mathring{p}_{\beta})|(q_{\alpha} -q_{\beta}) $;
	\item There exist a unique $ z \in \lbrack 0, \lcm(\mathring{P}^n) ) $ that satisfies the conditions $ \mathring{p}_i | (z -q_i) $,\ $ i \in \{1,2,...,n\} $.
\end{enumerate}
\label{3}
\end{theorem}

\begin{proof}
As an extension of Lemma 1, this theorem only needs to be verified its sufficiency from proposition 1) to 2).

From Lemma 1, if proposition 1) holds, there must exist $ z' \in ( -\infty, \infty) $ satisfying that $ \mathring{p}_i \mid (z' -q_i) $, $ i \in \{1,2,...,n\} $. Then, we set value $ z $ to $ z' \bmod \lcm(\mathring{P}^n) $. Through the divisibility theory, we have that $ \mathring{p}_i \mid ((z' \bmod \text{lcm}(P^n)) -q_i) $, that is, $ \mathring{p}_i \mid (z -q_i) $, where $ i \in \{1,2,...,n\} $ and $ z \in \lbrack 0, \lcm(\mathring{P}^n)) $. The existence of $ z \in \lbrack 0, \lcm(\mathring{P}^n) ) $ is proved.

Moreover, to verify the uniqueness of the above value $ z $, we assume that there exits two different $ z^{(1)}, z^{(2)} \in \lbrack 0, \lcm(\mathring{P}^n) ) $. They both satisfy the conditions in proposition 2). That is,
\begin{equation}
	\mathring{p}_i | (z^{(1)} -q_i),\ \mathring{p}_i | (z^{(2)} -q_i),\ i \in \{1,2,...,n\}.
\nonumber
\end{equation}
By the divisibility theory, we deduce that $ \mathring{p}_i \mid (z^{(1)} -z^{(2)}), i \in \{1,2,...,n\} $, which are further converted to $\text{lcm}(P^n) \mid (z^{(1)} -z^{(2)}) $. Then, owing to the definition $ z^{(1)}, z^{(2)} \in \lbrack 0, \lcm(\mathring{P}^n) ) $, the relationship $ | z^{(1)} -z^{(2)} | < \lcm(\mathring{P}^n) $ is confirmed. Hence, there must be $ (z^{(1)} -z^{(2)})  = 0 $, that is, $ z^{(1)} = z^{(2)} $. It violates the above assumption. The uniqueness of $ z $ holds.
\end{proof}

\begin{algorithm}[!t]
	\caption{Euclidean-based Slot Positioning}
	\KwIn{Conflict cliques: $ \{ \dot{\varpi}_{\kappa}^{\varepsilon_\iota} \}^{\varepsilon_\iota} \in \dot{\Upsilon} $ 
	}
	\KwOut{Re-scheduled sub-flows: $ \dot{f}_{\imath} \in \dot{\mathcal{F}} $
	
	}
	\For{$ \varepsilon_\iota $ in $ \mathcal{E} $}
	{
		\For{$ \dot{\varpi}_{\kappa}^{\varepsilon_\iota} $ in $ \dot{\Upsilon}^{\varepsilon_\iota} $}
		{
			$ \dot{q} \leftarrow 0 $ and $ \dot{p} \leftarrow 1 $\;
			\For{$ \phi^{\varepsilon_\iota} $ in $ \dot{\varpi}_{\kappa}^{\varepsilon_\iota} $}
			{	
				$ (q^{\varepsilon_\iota}, \mathring{p}) \leftarrow \phi^{\varepsilon_\iota} $\;
				$ x, y, \gcd(\dot{p}, \mathring{p}) \leftarrow $ EEA($ \dot{p}, \mathring{p} $)\;
				$ (\dot{q}, \dot{p}) \leftarrow $ update $ \dot{q} $ and $ \dot{p} $ by (21)\;
			}
			$ \{ f_i \} \leftarrow $ choose $ \phi^{\varepsilon_\iota} \in \dot{\varpi}_{\kappa}^{\varepsilon_\iota} $\;
			\For{$ f_i $ in $ \{ f_i \} $}
			{	
				$ \dot{f}_{\imath} \leftarrow $ get a sub-flow of $ f_i $ sharing attributes except $ \dot{p}_{\imath} \leftarrow \dot{p} $ and $ \dot{b}_{\imath} \leftarrow \dot{q}-(h_{i}^{\varepsilon_\iota}+o_i) $\;
				$ \dot{\mathcal{F}} \vartriangleleft \dot{f}_{\imath} $\; 
			}
		}
	}
\end{algorithm}

Due to the periodicity of sequences, it is similarly deduced that the unique $ z $ exists in each interval $ [k \cdot \lcm(\mathring{P}^n), (k+1) \cdot \lcm(\mathring{P}^n) ),\ k \in \mathbb{Z} $. These $ z $ are periodically aligned with period $ \lcm(\mathring{P}^n) $ and form a sequence as $ z + k \cdot \lcm(\mathring{P}^n) $. It formalizes the confluence slots $ \dot{s}^{\varepsilon_\iota} $ of conflict clique $ \dot{\varpi}_{\kappa}^{\varepsilon_\iota} \in \dot{\Upsilon} $. Further, the specific $ z $ is deduced and flow portions are backtracked by the Euclidean-based method shown in Algo. 5.

With the conflict cliques $ \dot{\Upsilon} $ traversal in lines 1-2, each of them $ \dot{\varpi}_{\kappa}^{\varepsilon_\iota} $ is positioned its confluence slots via the processes in lines 3-7. Specifically, the hyper-flows $ \phi^{\varepsilon_\iota} $ in clique $ \dot{\varpi}_{\kappa}^{\varepsilon_\iota} $ participate in confluences one by one. The initial confluence slots is formalized as $ \dot{q} + \dot{k} \cdot \dot{p} = 0 + \dot{k} \cdot 1 $ in line 3. With the confluence of $ \phi^{\varepsilon_\iota} $ in turn, the values of $ \dot{q} $ and $ \dot{p} $ are updated in lines 4-7. Within each turn, we aim to get the confluence slots $ \ddot{q} +\ddot{k} \cdot \ddot{p} $ of the existing $ \dot{q} + \dot{k} \cdot \dot{p} $ and newly joined $ \phi^{\varepsilon_\iota} $ with $ q_i^{\varepsilon_\iota} +k_i^{\varepsilon_\iota} \cdot \mathring{p}_i $, which are expressed as
\begin{equation}
	\{ \ddot{q} +\ddot{k} \cdot \ddot{p} \mid \dot{q} + \dot{k} \cdot \dot{p} = q_i^{\varepsilon_\iota} + k_i^{\varepsilon_\iota} \cdot \mathring{p}_i,\ \dot{k}, k_i^{\varepsilon_\iota} \in \mathbb{Z} \}.
	\nonumber
\end{equation}
Based on the above inferences, period $ \ddot{p} $ is given as $ \lcm(\dot{p}, \mathring{p}_i) $. Then, considering the uniqueness of $ \ddot{q} $ from Theorem 1, only a pair of specific $ \tilde{x} , \tilde{y} $ satisfying
\begin{equation}
	\dot{q} +\tilde{x} \cdot \dot{p} = q_i^{\varepsilon_\iota} +\tilde{y} \cdot \mathring{p}_i 
\end{equation}
is required to get the value $ \ddot{q} $ with $ (\dot{q} +\tilde{x} \cdot \dot{p}) \bmod \ddot{p} $. By the extended Euclidean algorithm (EEA), a pair of feasible $ x, y $ satisfying that $ x \cdot \dot{p} +y \cdot \mathring{p}_i = \gcd(\dot{p}, \mathring{p}_i) $ is given to convert (18) into 
\begin{equation}
	\dfrac{\tilde{x} \cdot \dot{p} -\tilde{y} \cdot \mathring{p}_i}{x \cdot \dot{p} + y \cdot \mathring{p}_i} = \dfrac{q_i^{\varepsilon_\iota} -\dot{q}}{\gcd(\dot{p}, \mathring{p}_i)}.
\end{equation}
Further, the feasible $ \tilde{x} $ and $ \tilde{y} $ in (19) are obtained by enhancing the above equiproportional relation as
\begin{equation}
	\dfrac{\tilde{x} \cdot \dot{p}}{x \cdot \dot{p}} = \dfrac{-\tilde{y} \cdot \mathring{p}_i}{y \cdot \mathring{p}_i} = \dfrac{q_i^{\varepsilon_\iota} -\dot{q}}{\gcd(\dot{p}, \mathring{p}_i)},
\end{equation}
which gives the specific value $ \tilde{x} $ as $ x \cdot (q_i^{\varepsilon_\iota} -\dot{q}) / \gcd(\dot{p}, \mathring{p}_i) $ and $ \tilde{y} $ satisfying (18), (19). Thus, the baseline $ \ddot{q} $ of confluence slots is obtained. It and period $ \ddot{p} $ together identify the confluence slots of the existing $ \dot{q} + \dot{k} \cdot \dot{p} $ and newly joined $ q_i^{\varepsilon_\iota} +k_i^{\varepsilon_\iota} \cdot \mathring{p}_i $ as the tuple
\begin{equation}
	(\ddot{q}, \ddot{p}) = ((\dot{q}_i +\tilde{x} \cdot \dot{p}_i) \bmod \ddot{p}, \lcm(\dot{p}, \mathring{p}_i)).
\end{equation}

With the above way performed cyclically in lines 6-7 , the confluence slots $ \dot{s}^{\varepsilon_\iota} $ of conflict clique $ \dot{\varpi}_{\kappa}^{\varepsilon_\iota} $ are precisely positioned without violently retrieving all occupied slots of flows. And then, we follow Theorem 2 to choose any one hyper-flow $ \phi^{\varepsilon_\iota} $ with flows $ { f_i } $ in this conflict clique. The re-scheduled flow portions $ \{ f_i^{\dot{\varpi}} \} $ are backtracked in line 8. Since they are the sub-sequences of flows $ { f_i } $, these flow portions are identified as sub-flows $ \dot{f}_{\imath} \in \dot{\mathcal{F}} $ in lines 9-10, where each of them inherits the flow $ f_i $ attributes except the period $ \dot{p}_{\imath} $ and baseline time of generation $ \dot{b}_{\imath} $. All backtracked sub-flows are aggregated as $ \dot{\mathcal{F}} $ without duplication in line 11.

\subsubsection{Low-cost sub-flow re-scheduling}
After backtracking the sub-flows $ \dot{\mathcal{F}} $, their fine-tuning is needed to avoid slot overflow and rectify the preliminary solution $ \mathcal{O} $ as slightly as possible. Hence, we design a low-cost HFG method as a variant of Algo. 2 to re-schedule sub-flows $ \dot{\mathcal{F}} $. It follows the incremental and hyper-flow graph based pattern, but with a lower computation cost to reduce the complexity.

As shown in Algo. 6, the sub-flows $ \dot{f}_{\imath} \in \dot{\mathcal{F}} $ are re-scheduled one by one in line 1. Each of them is allocated with the feasible slots by the offset $ \dot{o} $ search in lines 3-7. Relative to the whole flow, the new offset of sub-flow $ \dot{f}_{\imath} $ causes the flow jitter limited by $ \dot{j}_{\imath} $ in jitter constraint (5). Hence, the offset search range is bounded from $ (o_i-\dot{j}_{\imath})^+ $ to bound $ \bar{o}_{\imath} $ in line 2. For each offset $ \dot{o} $, its corresponding maximum slot occupancy $ \bar{\zeta}_o $ is obtained in line 4 based on the hyper-flow graphs and maximal cliques $ \Upsilon $. The specific processes follow lines 6-16 of Algo. 2. To avoid slot overflow, the feasible offset is filtered with condition $ \bar{\zeta}_o \leq \Lambda $. The first found among them is given
as the optimal offset $ \dot{o}_{\imath} $ in lines 5-6 to reduce scheduling complexity. If there is no feasible offset, the whole scheduling fails and stops in lines 8-9. Otherwise, the sub-flows with optimal offsets are stored as the supplemental solution $ \dot{\mathcal{O}} $ in line 10. The hyper-flows $ \Phi $, their weights $ \mathcal{L} $ and maximal cliques $ \Upsilon $ are updated via lines 24-29 of Algo. 2.

With this low-cost method, slot overflow is avoided and then the schedulability and load-balancing are improved. At the end of GH$^2$, the supplemental solution $ \dot{\mathcal{O}} $ and preliminary $ \mathcal{O} $ are aggregated as a complete scheduling solution.

\begin{algorithm}[!t]
	\caption{Low-cost HFG Scheduling Method}
	\KwIn{Re-scheduled sub-flows: $ \dot{f}_{\imath} \in \dot{\mathcal{F}} $ \\
		Sub-flow attributes: $ \{ \dot{l}_{\imath}, \dot{p}_{\imath}, \dot{b}_{\imath}, \dot{d}_{\imath}, \dot{j}_{\imath}, \dot{R}_{\imath} \} \in \dot{\mathcal{A}} $ \\
		Available slot capacity: $ \Lambda $ \\
		Global maximal cliques: $ \{ \varpi_{\kappa}^{\varepsilon_\iota} \}^{\varepsilon_\iota} \in \Upsilon $
	}
	\KwOut{Re-scheduling solution: $ (\dot{f}_{\imath}, \dot{o}_{\imath}) \in \dot{\mathcal{O}} $
	}
	\For{$ \dot{f}_{\imath} $ in $ \dot{\mathcal{F}} $}
	{	
		$ \dot{o} \leftarrow (o_i-\dot{j}_{\imath})^+ $ and $ \bar{o}_{\imath} \leftarrow \min \{ o_i+\dot{j}_{\imath}, \bar{o}_i \} $\;
		\While{$ \dot{o} < \bar{o}_{\imath} $}
		{
			$ \bar{\zeta}_o, \{ \phi^{\varepsilon_\iota} \}, \{ \tilde{\Upsilon}_{\phi}^{\varepsilon_\iota} \} \leftarrow $ follow line 6-16 in Algo. 2\;
			\If{$ \bar{\zeta}_o \leq \Lambda $}
			{
				$ \dot{o}_{\imath} \leftarrow \dot{o} $ and \textbf{break}\;
			}
			$ \dot{o} \leftarrow \dot{o} +1 $\;
		}
		\If{$ \dot{o}_{\imath} $ not exists}
		{
			Scheduling fails and \textbf{break}\;
		}
		$ \dot{\mathcal{O}} \vartriangleleft (\dot{f}_{\imath}, \dot{o}_{\imath}) $ and update $ \Phi, \mathcal{L} $ and $ \Upsilon $ following line 24-29 in Algo. 2\;
	}
\end{algorithm}

\newtheorem{remark}{Remark}
\begin{remark}
With the flow graph \cite{9714183}, the CCR re-scheduling methods also apply by only adjusting the process object from hyper-flow to flow nodes.
\label{1}
\end{remark}
\begin{remark}
The flow fine-tuning phase is actually the decomposition and re-scheduling of flows resulting in slot overflow. For these sub-flows that fail in the above re-scheduling due to slot overflow, their further decomposition as a new re-scheduling round could better balance the load. Ignoring the constraints of latency and jitter, the CCR re-scheduling round by round tends to balance the slot loads ideally.
\label{2}
\end{remark}
\begin{remark}
The scheduling procedure of GH$^2$ is also suitable for online and dynamic network scenarios, where the time-sensitive applications or network structures evolve over time. Once that occurs, GH$^2$ could cleanup silent application flows from previous scheduling solution $ \mathcal{O} $ and slot occupancy $ \Upsilon $. Then, newly joined flows are scheduled by Algo. 1 to update the current solution. In this way, the on-the-fly network system is re-configured efficiently without affecting in-transit flows.
\label{3}
\end{remark}

\section{Simulation and Evaluation}
In this section, we conduct a comprehensive analysis of GH$^2$ scheduling algorithm for its scalability in terms of schedulability, scheduling efficiency, real-time transmission ability and load balancing ability.

\subsection{Simulation Setup}
All the simulations are run in the Intel (R) Eight-Core (TM) i7-9700HQ CPU with 3.0 GHz and 16 GB of RAM.

\subsubsection{Network Parameters}
Compared to CEV network \cite{10101832, 8700610} with more switches, simulation networks are constructed with 8 TSN switches configured statically by the CQF model. It facilitates the network simulation with more intensive flow confluence to validate scheduling performances. The specific switch topologies are networked as the linear, ring and tree forms shown in Fig. \ref{fig6}, respectively. These topologies could assemble almost all structures of industrial networks. With the typical settings of 125 us slot length $ T $, $<$2 us synchronization errors $ \delta^{\varepsilon_\iota} $, 1 Gbit/s link bandwidth $ \Gamma $, 1 Mb queue depth $ \mho^{\varepsilon_\iota} $ and 80\% distribution factor $ \gamma $, the available slot capacity $ \Lambda $ is given to 12.3 KB.

\begin{figure}[!h]
	\centerline{\includegraphics[width=7cm]{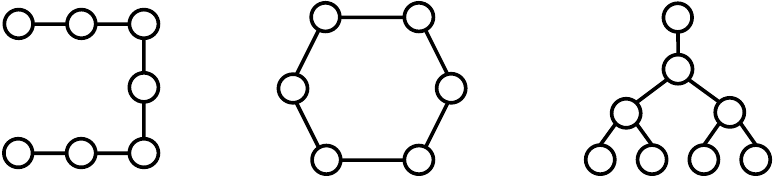}}
	\caption{Schematic of three simulation switch topologies}
	\label{fig6}
\end{figure}

\subsubsection{Flow Attributes}
Referring to the TSN profile standard for industrial automation \cite{9714184}, we generate simulation flow cases that close to actual industrial scenarios with standardized attributes illustrated in Table \ref{tab2}.

\begin{figure*}[!h]
	\centering
	\subfloat[\label{fig:a}]{
		\hspace{-0.9cm}
		\includegraphics[width=6.4cm]{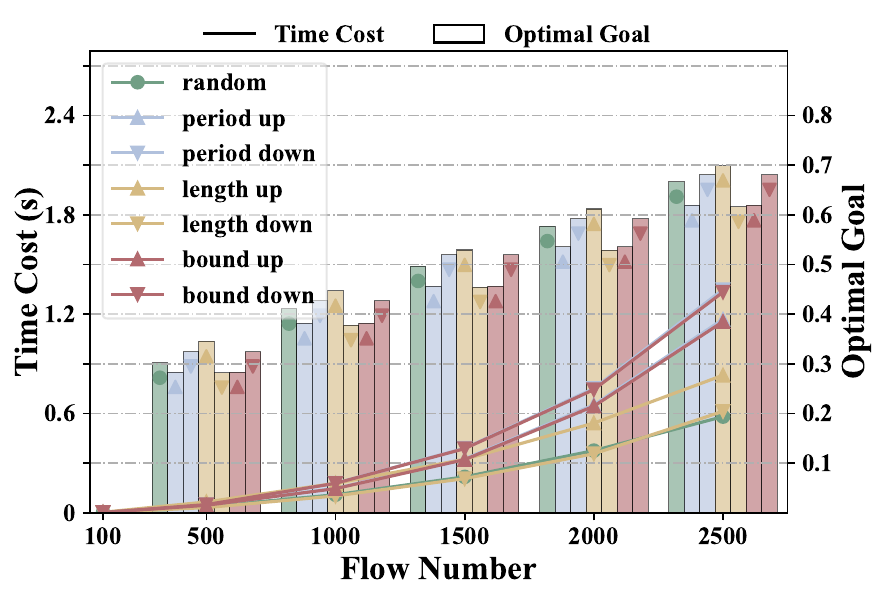}}
	\subfloat[\label{fig:b}]{
		\includegraphics[width=6.4cm]{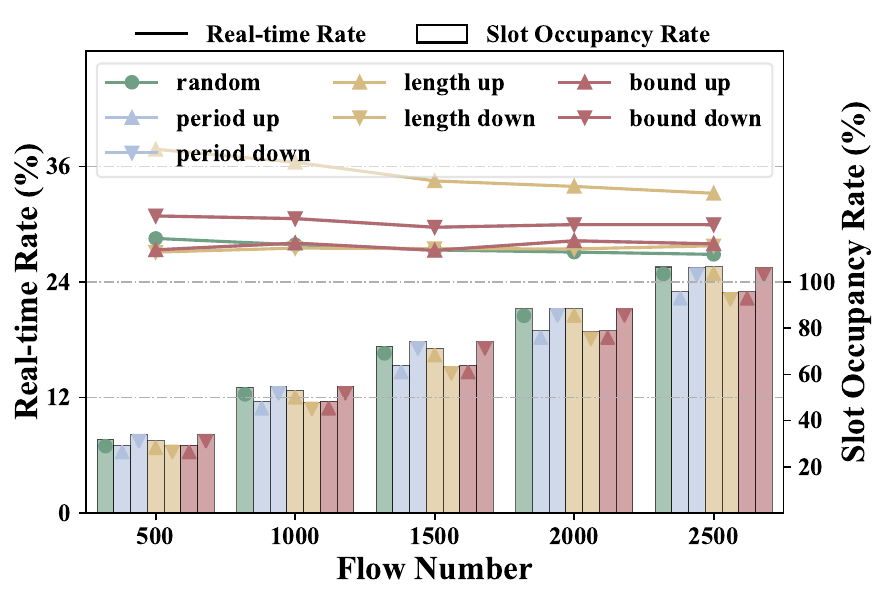}}
	\subfloat[\label{fig:c}]{
		\includegraphics[width=6.4cm]{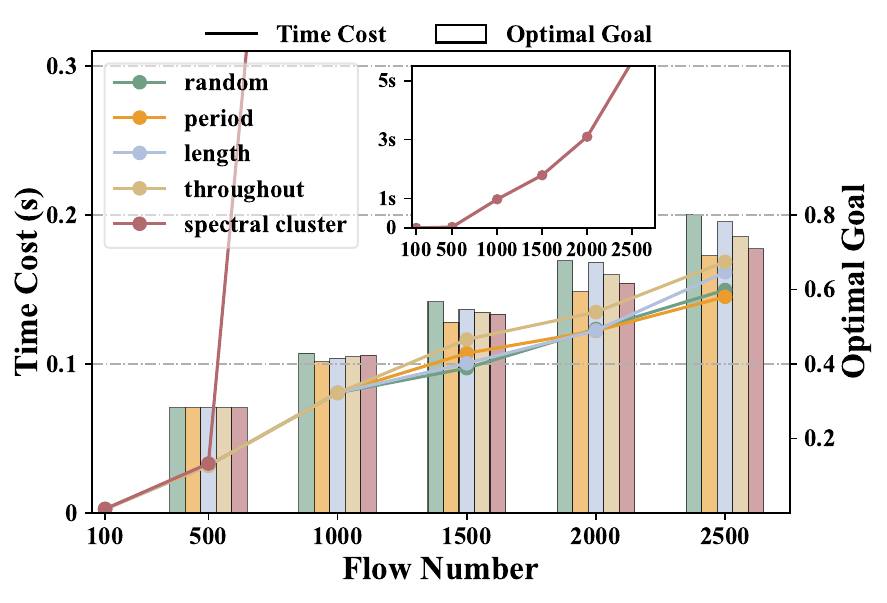}}
	
	\subfloat[\label{fig:a}]{
		\hspace{-0.9cm}
		\includegraphics[width=6.4cm]{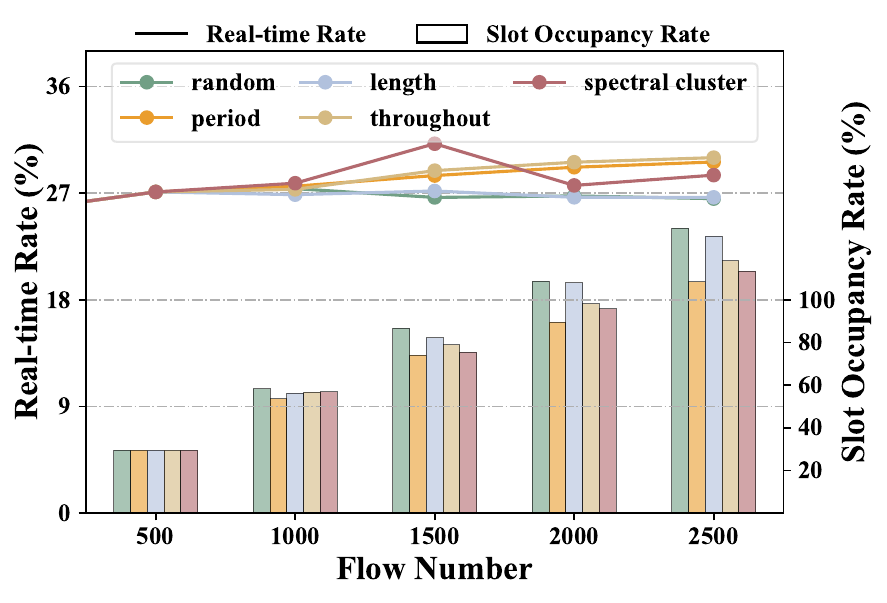}}
	\subfloat[\label{fig:b}]{
		\includegraphics[width=6.4cm]{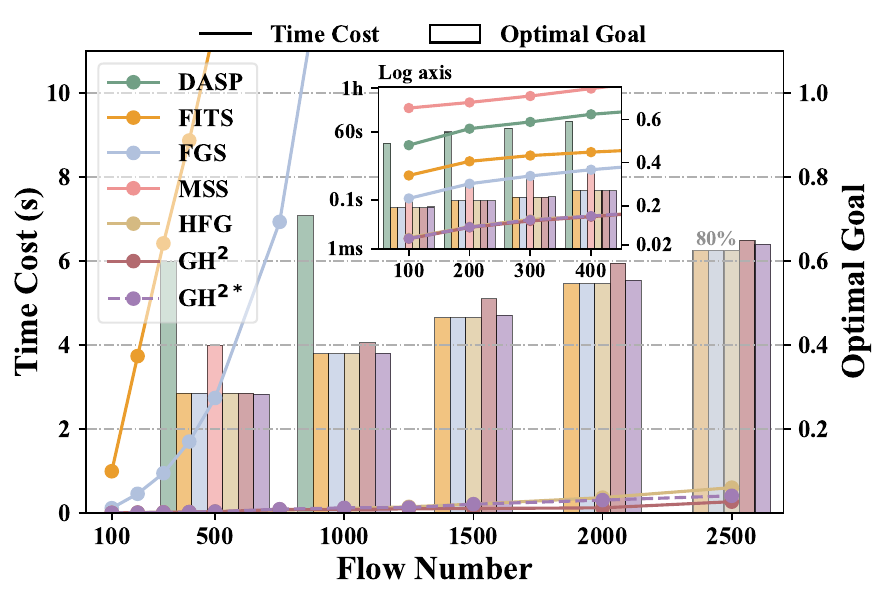}}
	\subfloat[\label{fig:c}]{
		\includegraphics[width=6.4cm]{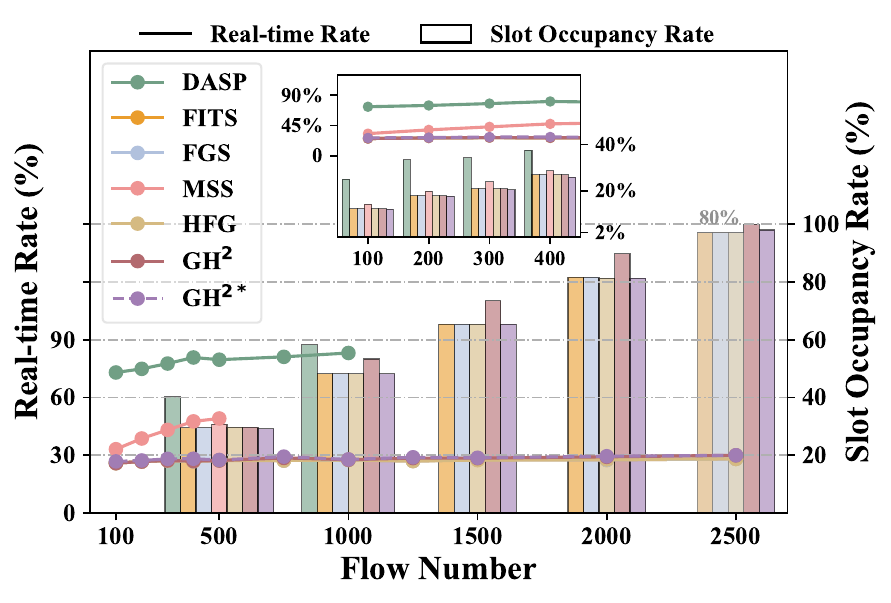}}
	
	\caption{(a) \& (b) Sorting strategy comparison, (c) \& (d) Partition strategy comparison (e) \& (f) Algorithm comparison}
	\label{simulation1}
\end{figure*}

\begin{table}[!t]
	\caption{Flow Attributes for Simulation}
	\renewcommand{\arraystretch}{1.25}
	\setlength{\tabcolsep}{15pt}
	\begin{center}
		\begin{tabular}{p{70pt}|p{90pt}}
			\hline
			Attribute& 
			Value range of flow cases $^{\mathrm{a}}$ \\
			\hline
			frame length $ l_i $& 
			64 Bytes $ \thicksim $ 1500 Bytes \\
			flow period $ p_i $& 
			Type 1: 1 ms $ \thicksim $ 200 ms \par Type 2: 1 ms $ \thicksim $ 400 ms \par Type 3: 1 ms $ \thicksim $ 600 ms \par Type 4: 1 ms $ \thicksim $ 800 ms \\
			basetime $ b_i $& 
			0 ms $ \thicksim $ $ p_i $ ms \\
			allowable latency $ d_i $& 
			0.1$ \cdot p_i $ ms $ \thicksim $ 0.5 $ \cdot p_i $ ms \\
			allowable jitter $ j_i $& 
			500 us $ \thicksim $ 0.1$ \cdot p_i $ ms \\
			hop count $ h_i $& 
			1 hop $ \thicksim $ 6 hops \\
			\hline
			\multicolumn{2}{p{160pt}}{$^{\mathrm{a}}$Both unicast and multicast flows are considered.}
		\end{tabular}
	\end{center}
	\label{tab2}
\end{table}

\subsection{Simulation Results}
To obtain accurate and stable scheduling performances, we repeat 20 tests for each simulation with different flow cases. The final results are given by averaging. Except for specific parameter evaluations, the optimization weight factor $ \rho $ is set as 0.5 to balance the average real-time transmission rate and global maximum slot occupancy rate (abbreviated as real-time rate and slot occupancy rate later). The partition scale $ \Xi $ is set as 500 flows to balance the scheduling time and optimality. Moreover, the normal simulations adopt the linear topology and flow cases with 1ms $ \thicksim $ 200ms period interval.

\subsubsection{Sorting Strategy Comparison}
In incremental scheduling of the HFG method for each partition, different flow orders are compared to filter the optimal sorting strategy, including the random, down or up order by flow period, frame length and offset bound \cite{10001004}, respectively. As shown in Fig. \ref{simulation1}(a)\&(b), the length descent sorting reflects the fastest scheduling efficiency and best optimization capability. Separately, both the real-time rate and slot occupancy rate of it rank at a relatively low echelon compared to the other sorting strategies. Thus, we adopt the down order of length as the optimal sorting strategy during the following simulations.

\subsubsection{Partition Strategy Comparison}
This simulation works to compare the different partition methods and filter the optimal one. The flow cases are partitioned by the flow period, frame length, bandwidth consumption illustrated in (9) and the SOTA randomly generating \cite{8342096}, spectral cluster \cite{8889667}, respectively. These flow partitions are scheduled via the first three phases of GH$^2$. As reflected in Fig. \ref{simulation1}(c), the period-based flow partitioning shows relatively superior performances under different flow scales in terms of the time cost and optimal goal. Fig. \ref{simulation1}(d) shows that the period-based partitioning is superior at slot occupancy rate, while length-based partitioning benefits the real-time rate. Synthesizing the above results, we validate the benefits of attribute-driven flow partitioning and adopt the period-based strategy for GH$^2$.

\subsubsection{Re-scheduling method Comparison}
In the flow fine-tuning phase, the systematic CCR re-scheduling method can be performed with the hyper-flow or flow graph based pattern. To evaluate their re-scheduling capability and then filter the optimal way, we decrease the available slot capacity $ \Lambda $ in gradient under 2000 flows. It starts from the slot occupancy rate after the first three phases for GH$^2$. As shown in Fig. \ref{simulation2}, with the decreasing of slot capacity, the re-scheduling time grows due to the increasing conflict cliques and re-scheduled sub-flows, while the real-time rate remains almost invariant owing to precisely fine-tuning. Comparing the above two ways, we adopt the faster hyper-flow graph based CCR method for GH$^2$.
\begin{figure}[!t]
	\centerline{\includegraphics[width=6.8cm]{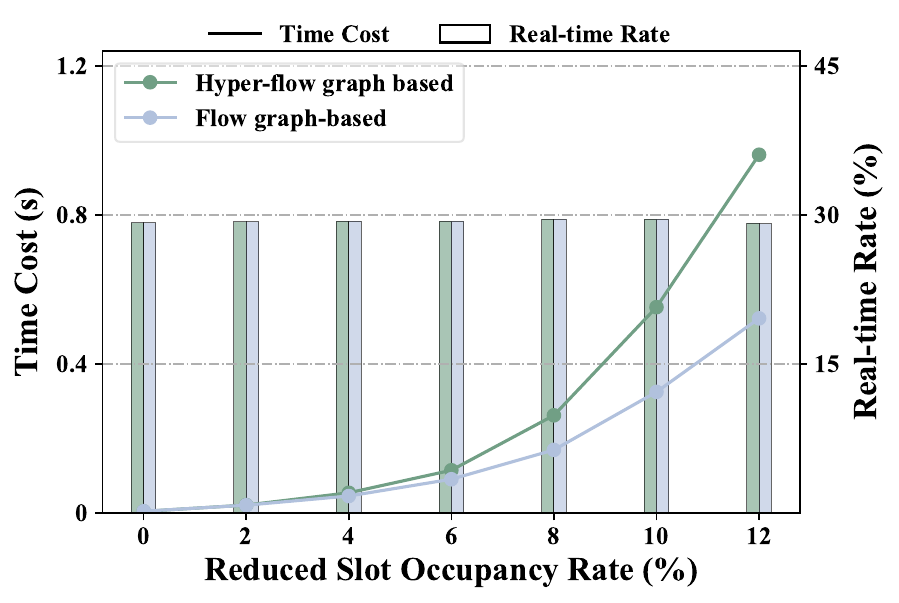}}
	\caption{Re-scheduling method comparison}
	\label{simulation2}
\end{figure}

\subsubsection{Algorithm Comparison}
Adopting the above strategies, our proposed GH$^2$ algorithm is compared its scheduling performances with SOTAs as DASP \cite{8889667}, FITS \cite{9407828}, FGS \cite{9714183}, MSS \cite{9893358} for DSP. They pursue the same optimization goal (8). As shown in Fig. \ref{simulation1}(e)\&(f), our proposed HFG method reflects its Pareto optimization including higher scheduling efficiency, more balanced slot loads and lower transmission latency, while its slot occupancy rate and real-time rate are equal to these of FITS and FGS since their same offset filtering condition. Furthermore, GH$^2$ improves the scheduling efficiency and schedulability at different flow scales, especially in large-scale complex flow scenarios. Under 1000 flows, its runtime is reduced to nearly 1/300, 1/200 of the SOTA FITS, FGS algorithms, respectively. Under 2000 flows, the ratios are 1/430, 1/1000. Besides, when the flow scale grows to 2500, the SOTA algorithms are hard to get a feasible solution. GH$^2$ improves 80\% schedulability of FITS and FGS methods to 100\%. Although the load balancing ability is sacrificed due to the separate partition scheduling that restricts the global optimality, GH$^2$ can improve it with the CCR re-scheduling method, where the results are marked as GH$^{2*}$. In this way, GH$^2$ constructs the Pareto front of scheduling efficiency and load balancing, and verifies its better scalability in various flow scenarios.

Moreover, three scheduling patterns shown in Fig. \ref{fig1_3} and adopted in FITS, FGS, HFG are compared for their computing complexity with the scheduling information scale (abbreviated as Sch-Info scale), respectively. Fig. \ref{simulation3} reflects that the hyper-flow graph based and flow graph-based scheduling patterns compress the slot occupancy information of the frame-based pattern exponentially. It is because of the desensitization of flow attributes. Compared to the flow graph-based pattern, the hyper-flow way shrinks the graph scale significantly. These explain and confirm the scheduling efficiency improvement of our HFG and GH$^2$ algorithms. Then, the scheduling efficiency stability of GH$^2$ compared to SOTAs is evaluated under different flow scales. As shown in Fig. \ref{simulation4}, GH$^2$ reflects the least time fluctuation and significant stability improvement. When more re-scheduling occurs as GH$^{2*}$, its stability is relatively reduced due to the scale uncertainty of re-scheduled sub-flows, but still superior to these SOTA algorithms.
\begin{figure}[!t]
	\centerline{\includegraphics[width=6.4cm]{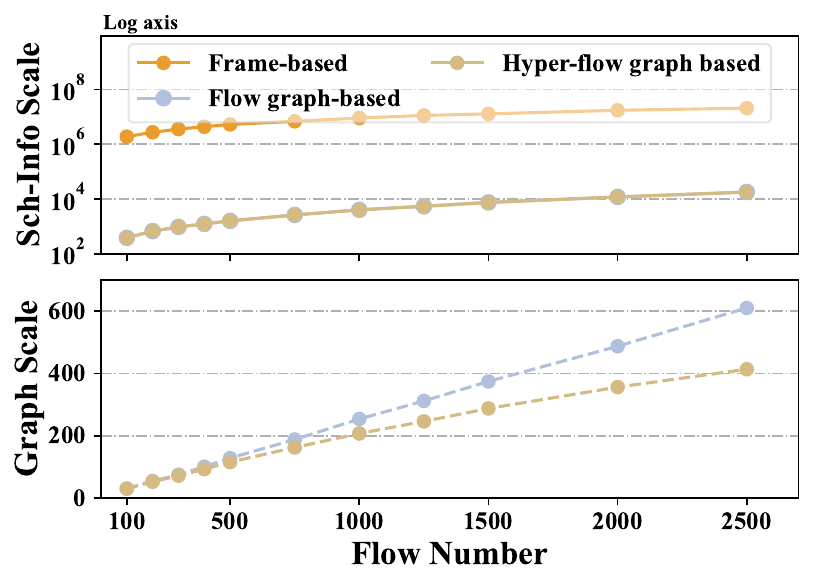}}
	\caption{Scheduling pattern comparison}
	\label{simulation3}
\end{figure}

\begin{figure}[!h]
	\centerline{\includegraphics[width=6.4cm]{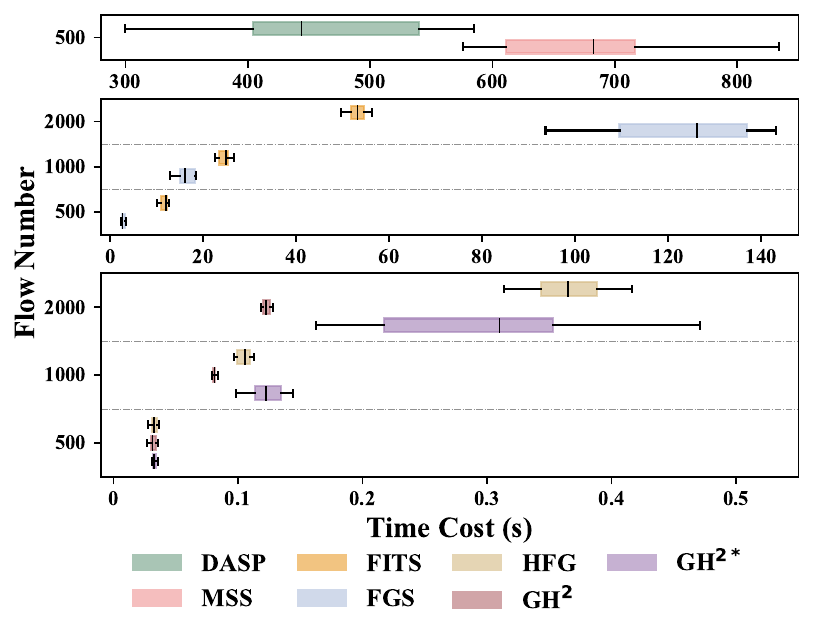}}
	\caption{Algorithm stability comparison}
	\label{simulation4}
\end{figure}

\begin{figure*}[!t]
	\centering
	\subfloat[\label{fig:a}]{
		\hspace{-0.95cm}
		\includegraphics[width=6.4cm]{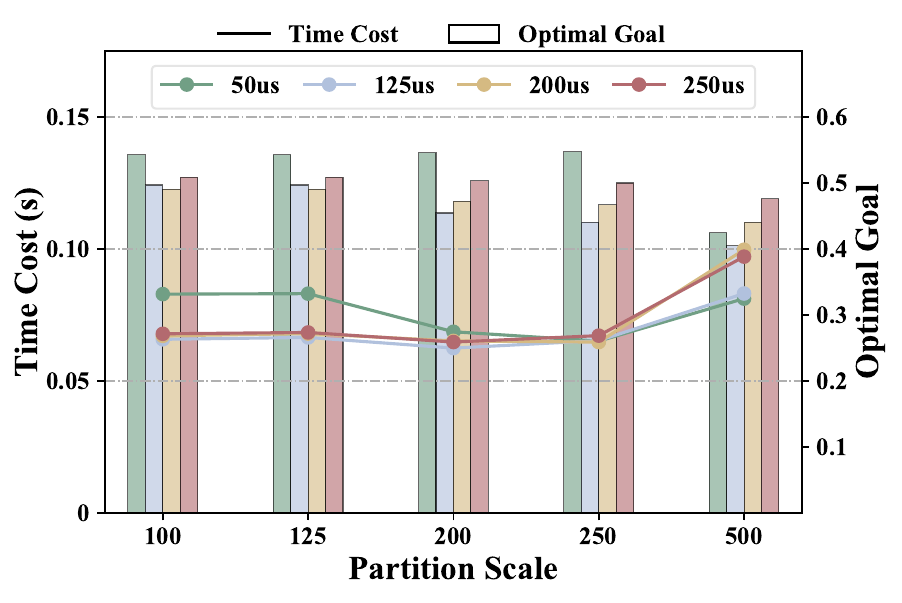}}
	\subfloat[\label{fig:b}]{
		\includegraphics[width=6.4cm]{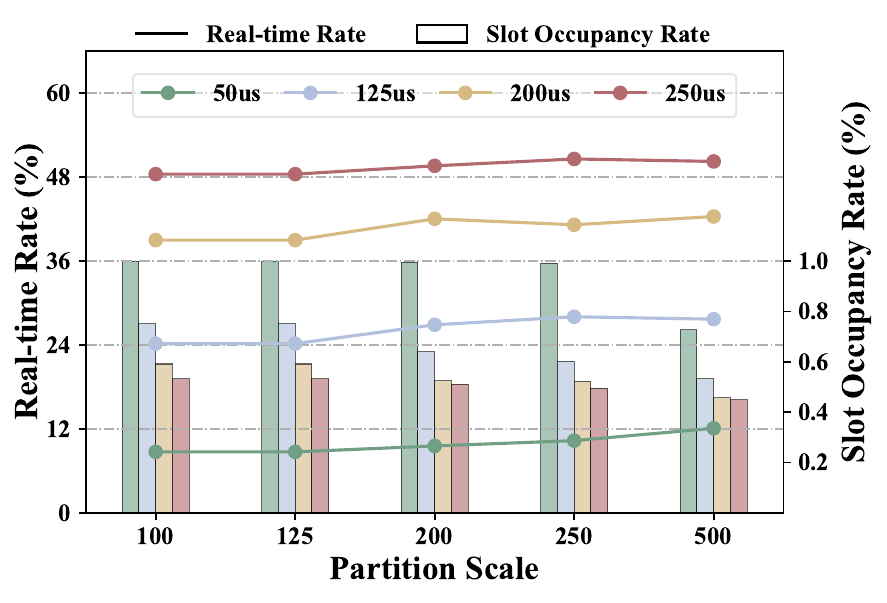}}
	\subfloat[\label{fig:c}]{
		\includegraphics[width=6.4cm]{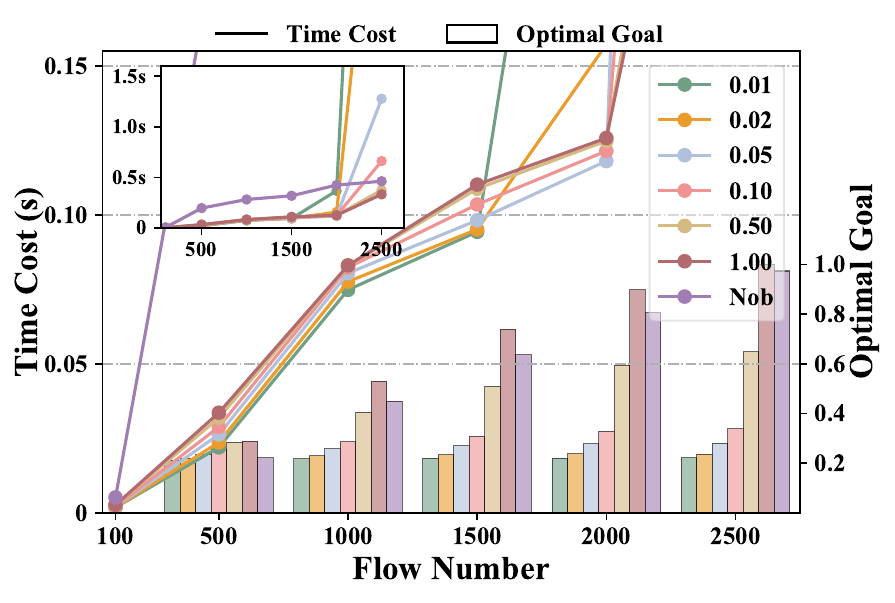}}
	
	\subfloat[\label{fig:a}]{
		\hspace{-0.94cm}
		\includegraphics[width=6.4cm]{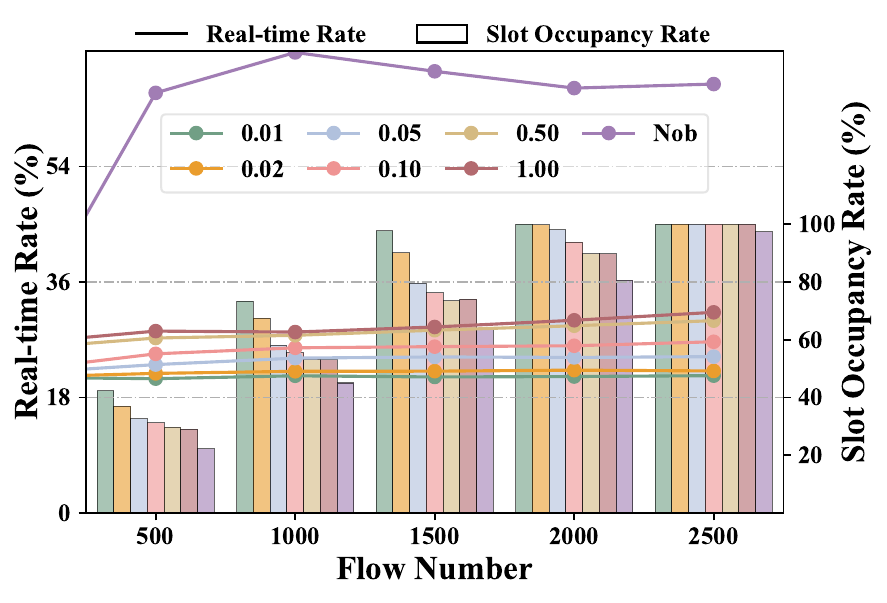}}
	\subfloat[\label{fig:b}]{
		\includegraphics[width=6.4cm]{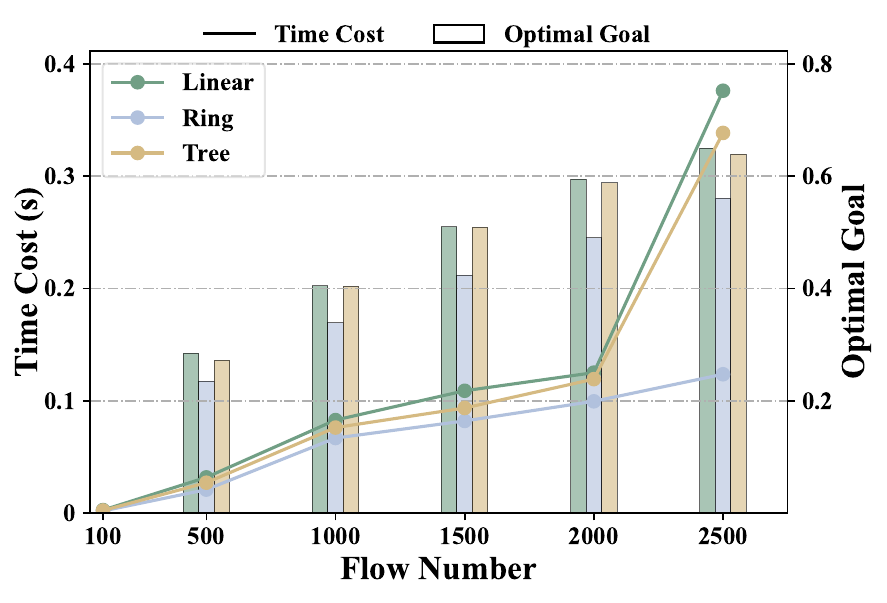}}
	\subfloat[\label{fig:c}]{
		\includegraphics[width=6.4cm]{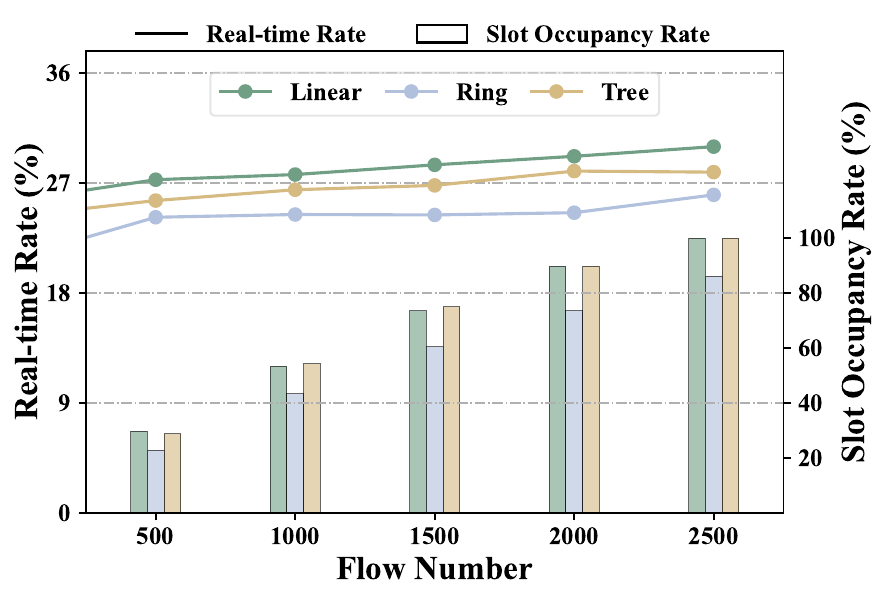}}
	
	\caption{(a) \& (b) Partition scale and slot length evaluation, (c) \& (d) Optimization weight factor evaluation, (e) \& (f) Network topology influence}
	\label{simulation5} 
	\renewcommand{\thefigure}{8}
\end{figure*}

\subsubsection{Partition Scale \& Slot length Evaluation}
In GH$^2$, the key parameters containing partition scale $ \Xi $ and slot length $ T $ are evaluated under 1000 flows for their influences on scheduling performances. From Fig. \ref{simulation5}(a)\&(b), we conclude that both the increasing partition scale and slot length contribute to improving load balancing while raising the transmission latency. Compared to the partition scale, the slot length shows a greater impact. It allows them to work together for better performances. For the scheduling efficiency, it is influenced conjointly by these two parameters. With their synergizing, GH$^2$ consistently schedules 1000 flows within 100 ms.

\subsubsection{Optimization Weight Factor Evaluation}
Just like above, we evaluate the influences of optimization weight factor $ \rho $ for GH$^2$ at different flow scales. Besides different weight factors, another local optimization goal from \cite{9714183} is evaluated with the expression as $ \min_{o_{\epsilon}}{ \frac{\tilde{\zeta}_{\epsilon}}{\Lambda} } $ and marked as $ \text{Nob} $. $ \tilde{\zeta}_{\epsilon} $ is the maximum occupancy among slots containing current scheduled flow $ f_{\epsilon} $ and not applied for the early-break strategy. The simulation results are shown in Fig. \ref{simulation5}(c)\&(d). As the weight factor grows, the scheduling time and real-time rate increase while the slot occupancy rate decreases. Hence, the factor $ \rho $ is conformed for its effect on balancing the above two optimization goals, which extend the Pareto front constructed by GH$^2$. Specially, when flows reach a certain scale, slot overflow occurs and a relatively higher time cost is consumed by the re-scheduling. The smaller the factor $ \rho $, the earlier this occurs. Moreover, the $ \text{Nob} $-directed scheduling shows the highest time cost, real-time rate and lowest slot occupancy, which suits the scenarios with stronger load balancing preferences.

\subsubsection{Network Topology Influence}
To verify the universality of GH$^2$, we evaluate its performances in linear, ring and tree topologies, respectively, where different simulation flow cases share the same attributes except their topologies. As shown in Fig. \ref{simulation5}(e)\&(f), GH$^2$ manifests the low scheduling complexity and superior real-time transmission \& load balancing abilities in all these topologies. Among them, the ring topology shows better performances since its relatively short route from the source to the destination host.

\subsubsection{Period Attribute Influence}
Considering the period influence discussed in Section III, we evaluate the scheduling stability of GH$^2$ under different flow period intervals. Four types of flow cases as shown in Table \ref{tab2} are scheduled at different flow scales. Fig. \ref{simulation6}(a)\&(b) reflect that all of them have low scheduling runtime and similar optimization performances, where the extension of period interval somewhat raises the scheduling efficiency and improves both the real-time transmission and load balancing abilities. It further confirms the scalability for various flow scenarios and desensitization for flow attributes of GH$^2$.
\begin{figure}[!t]
	\centering
	\subfloat[\label{fig:a}]{
		\includegraphics[width=6.4cm]{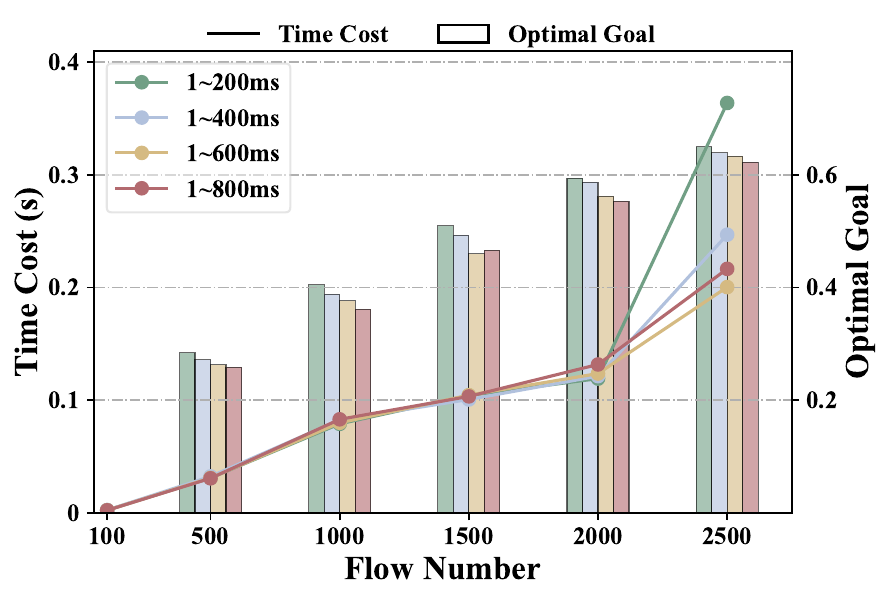}}
	
	\subfloat[\label{fig:a}]{
		\includegraphics[width=6.4cm]{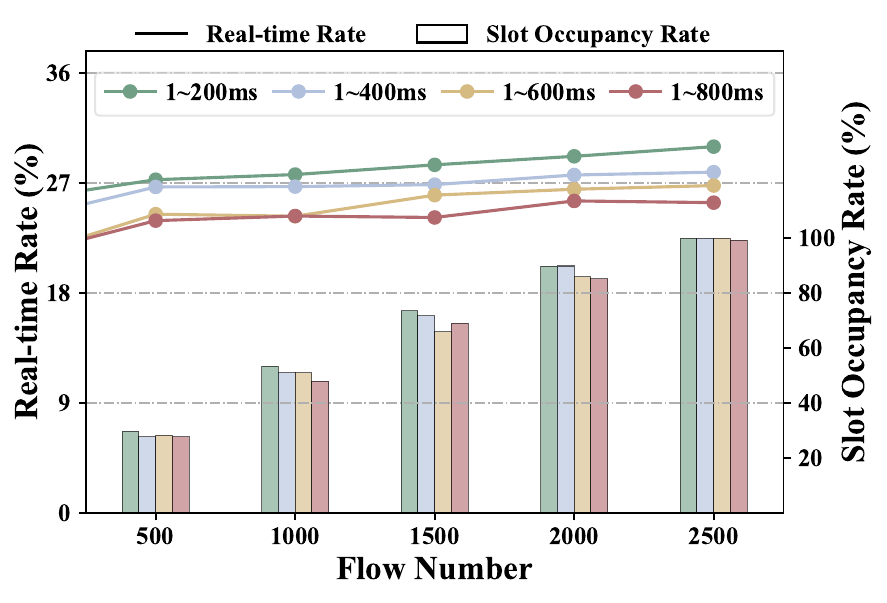}}
	\caption{(a) \& (b) Period attribute influence}
	\label{simulation6} 
\end{figure}

\section{Conclusion}
In this paper, we investigate the scalable scheduling for the CQF-based DSP with large-scale complex industrial flows. Considering its high complexity and limited optimality, we deeply mine the attribute-driven scheduling features and develop a hyper-flow graph based scheduling scheme. The hyper-flow graph is built by taking the similar flow sets as the hyper-flow nodes and proved its equivalence with flow attribute-sensitive scheduling and re-scheduling information. It embeds the redundant scheduling information in less maximal cliques, and precisely reverse maps them to overflow flow portions. Under the systematical guidance of the above methodology, a “flow division-conquering-synthesis-fine tuning” hierarchical framework is built to optimize and balance the scheduling complexity and optimality. Its parallel scheduling reduces the device and flow scale induced complexity, while the precise fine-tuning improves the schedulability and load balancing. Further, this scheme is refined as a comprehensive scheduling algorithm GH$^2$, including the attribute-driven lightweight partitioning, parallel HFG scheduling, parallel Bron-Kerbosch synthesizing and precise CCR re-scheduling methods. Simulation results demonstrate the superiority of GH$^2$ in terms of the scheduling efficiency, schedulability and QoS performances along a Pareto front. This scheme can also be extended into the online and dynamic network or other time-related scheduling scenarios.

\bibliographystyle{ieeetr}
\bibliography{dscp.bib}

\clearpage

\vfill

\end{document}